\documentclass{article}

\PassOptionsToPackage{numbers, compress}{natbib}
\usepackage[preprint]{neurips_2023}

\usepackage[utf8]{inputenc} % allow utf-8 input
\usepackage[T1]{fontenc}    % use 8-bit T1 fonts
\usepackage{hyperref}       % hyperlinks
\usepackage{url}            % simple URL typesetting
\usepackage{booktabs}       % professional-quality tables
\usepackage{amsfonts}       % blackboard math symbols
\usepackage{nicefrac}       % compact symbols for 1/2, etc.
\usepackage{microtype}      % microtypography
\usepackage[table,xcdraw]{xcolor}         % colors
\usepackage{amsmath,amsfonts,amssymb,theorem,euscript,array,enumerate,amsfonts,mathrsfs}
\usepackage{array}
\usepackage{float}
\usepackage{bbm}
\usepackage{paralist}
\usepackage{cleveref}
\usepackage{physics}
\usepackage{algorithm}
\usepackage{algorithmic}
\usepackage{graphics} 
\usepackage{graphicx}
\usepackage{multirow}
\graphicspath{{./figures/}}
\usepackage[font=footnotesize,labelfont=bf]{caption}
\usepackage{subcaption}
\usepackage{enumitem,kantlipsum}
\usepackage{wrapfig,lipsum,booktabs}
\usepackage{hyperref}
\hypersetup{    
    urlcolor=cyan,
    pdftitle={Overleaf Example},
    pdfpagemode=FullScreen,
    }
\DeclareMathOperator{\argmin}{argmin}

\newcommand\name{\texttt{FedCore}}
\newcommand\fedavg{\texttt{FedAvg}}
\newcommand\fedavgds{\texttt{FedAvg-DS}}
\newcommand\fedprox{\texttt{FedProx}}

\newtheorem{Theorem}{Theorem}[section]

\newtheorem{Assumption}[Theorem]{Assumption}
\newtheorem{Lemma}[Theorem]{Lemma}

\newenvironment{proof}[1][{\it Proof.}]{\begin{trivlist}
\item[\hskip \labelsep {\bfseries #1}]}{ \hfill
$\Box$\end{trivlist}\vskip -0.2 cm}

\def \E{\mathbb{E}}

\def \Lc{{\cal L}}

\def \Oc{{\cal O}}

\def \Wc{{\cal W}}

\title{FedCore: Straggler-Free Federated Learning with Distributed Coresets}

\author{%
Hongpeng Guo$^{1}$  Haotian Gu$^{2}$ Xiaoyang Wang$^1$ Bo Chen$^1$\\
\textbf{Eun Kyung Lee}$^3$ \textbf{Tamar Eilam}$^3$ \textbf{Deming Chen}$^1$ \textbf{Klara Nahrstedt}$^1$\\
$^1$ UIUC \quad $^2$UC Berkeley \quad $^3$IBM Research\\
\texttt{\{hg5, xw28, boc2, dchen, klara\}@illinois.edu}\\
\texttt{haotian\_gu@berkeley.edu}\\
\texttt{\{eunkyung.lee, eilamt\}@us.ibm.com}
}

\begin{document}

\maketitle

\begin{abstract}
Federated learning (FL) is a machine learning paradigm that allows multiple clients to collaboratively train a shared model while keeping their data on-premise. However, the straggler issue, due to slow clients, often hinders the efficiency and scalability of FL. This paper presents \name{}, an algorithm that innovatively tackles the straggler problem via the decentralized selection of \emph{coresets}, representative subsets of a dataset. Contrary to existing centralized coreset methods, \name{} creates coresets directly on each client in a distributed manner, ensuring privacy preservation in FL. \name{} translates the coreset optimization problem into a more tractable k-medoids clustering problem and operates distributedly on each client. Theoretical analysis confirms \name{}'s convergence, and practical evaluations demonstrate an \textbf{8x} reduction in FL training time, without compromising model accuracy. Our extensive evaluations also show that \name{} generalizes well to existing FL frameworks\footnote{Code: \url{https://github.com/hongpeng-guo/FedCore}}.
% ------------------------------------------------------

% Federated Learning enables collaborative model training across multiple clients while preserving data privacy. However, its efficiency and scalability are often hindered by slow clients, known as stragglers. This paper introduces \name{}, an innovative algorithm designed to tackle this issue. Unlike existing coreset methods for centralized datasets, \name{} operates on each client, creating distributed coresets locally for added privacy. These coresets, representative data subsets, effectively reduce computational load, particularly for slower clients. \name{} dynamically identifies optimal coresets, adapting to evolving model parameters. It integrates into existing FL frameworks seamlessly, transforming coreset optimization into a tractable k-medoids clustering problem. \name{}'s convergence is theoretically confirmed, and evaluations highlight an \textbf{8x} training time reduction, maintaining model accuracy.

\vspace{-8pt}
\end{abstract}
\vspace{-8pt}
\section{Introduction}\label{sec:introduction}

Federated learning (FL) enables multiple clients to collaboratively train a shared machine learning model while retaining their data locally. It has greatly enhanced various privacy-sensitive domains by harnessing the power of AI and providing tailored solutions, including cancer diagnosis \cite{sheller2019multi, jimenez2021memory, li2019privacy}, urban transportation surveillance \cite{liu2020privacy, feng2020pmf}, financial services \cite{xu2019hybridalpha, long2020federated} and beyond. 
Federated learning has given rise to several research areas, including model convergence optimization \cite{mcmahan2017communication, xie2019asynchronous, li2019convergence}, FL system efficiency \cite{kim2021autofl, lai2021oort, guo2022bofl}, privacy preservation \cite{nasr2019comprehensive,bhowmick2018protection}, and robustness against adversarial attacks \cite{lyu2022privacy}. Among these areas, the straggler problem, caused by slow or unresponsive clients, hinders overall training efficiency and scalability. Meta's million-client FL system, Papaya \cite{huba2022papaya}, demonstrated that per-client training time distribution spans over two orders of magnitude, and the round completion time is 21x larger than the average training time per client due to stragglers' delays. Thus, efficient straggler mitigation is vital to unlock FL's full potential across diverse applications.

\paragraph{Motivations.} 

Existing solutions like client selection mechanisms \cite{nishio2019client, lai2021oort, kim2021autofl} and asynchronous frameworks \cite{xie2019asynchronous, nguyen2022federated, huba2022papaya, van2020asynchronous, li2021stragglers, chai2021fedat} aim to mitigate the straggler issue in federated learning (FL). However, these methods inherently treat the symptoms rather than the cause. Client selection can result in biased training \emph{data} due to the exclusion of slower clients, while asynchronous approaches can encounter staleness and inconsistency due to laggard updates from stragglers with slow \emph{hardware}.
These strategies don't  directly address the root cause of the straggler issue, which is due to the \emph{system and data volume heterogeneity} among clients in FL. The disparities in both computational capacity and data volume lead to varied training times, impacting overall efficiency.

Instead of sidestepping this fundamental challenge, our approach confronts it directly by \emph{aligning each client's data volume with its computational capability}. Recognizing that upgrading clients' hardware is impractical, we propose adjusting the amount of data processed by each slow client. These straggler clients often hold more data than can be efficiently processed within the allotted round time. To address this, we propose creating a representative \textbf{coreset}, a compact subset of the full dataset that encapsulates essential learning information. This strategy offers a more precise and direct solution to the straggler problem in FL.

In contrast to existing coreset generation solutions \cite{mirzasoleiman2020coresets, killamsetty2021grad}, where training data are collected on a central server to create a single coreset, we propose a \emph{distributed} approach that forms training coresets on each client independently, maintaining the \emph{privacy} integral to FL. This task is challenging, particularly when dealing with heterogeneous data distribution across numerous clients, each requiring different coreset sizes based on their computational capabilities.
Further complexity arises from the dynamic nature of machine learning models, which is constantly updated during the training process, necessitating the creation of adaptive coresets that can be adjusted according to different model parameters and training phases.
To tackle these issues, we designed \name{} which addresses two key \textbf{questions}:
\begin{enumerate}[leftmargin=*,label= ({\textbf{Q\arabic*}})]
\item How can we select statistically unbiased coresets that adapt to continuously updated models?
\item How to seamlessly integrate coreset generation with minimal overhead into FL frameworks?
\end{enumerate}

\paragraph{Methods and Results.} 
To generate statistically unbiased coresets that adapt to the evolving ML models, we design \name{}, which is applied independently to each client. \name{} operates by periodically searching for a coreset at the start of each FL round, ensuring that the selected coresets may differ between training rounds. This adaptability allows for the provision of the most suitable learning samples, taking into account the varying model parameters at different stages of training (\textbf{Q1}). 
To minimize coreset generation overhead, we employ gradient-based methods that leverage the per-sample gradients obtained during the gradient descent model training. By repurposing these gradients as input for our coreset algorithm, we optimize the use of available resources and eliminate the need for additional computations. Furthermore, we tackle the intricate coreset optimization problem by transforming it into a more manageable k-medoids clustering problem. This transformation allows for a more efficient resolution of the optimization task, streamlining the overall process and minimizing the system overhead (\textbf{Q2}).
Overall, this paper offers the following \textbf{contributions}:

\begin{enumerate}[leftmargin=*,label= ({\textbf{\arabic*}})]
\item We design and implement the \name{} algorithm, a pioneering solution that leverages distributed coreset training to address the straggler problem in FL with minimal system overhead.
\item We provide a theoretical convergence analysis for the \name{} algorithm, which manages to incorporate the coreset gradient approximation error with the federated optimization error, proving that federated model training with per-client coresets results in highly accurate models.
\item We extensively evaluate \name{} against existing solutions and baselines. Evaluation results indicate an 8x reduction in FL training time without degrading model accuracy compared to baseline \fedavg{}. In comparison to \fedprox{}, which handles stragglers through fewer local training epochs, \name{} consistently achieves faster convergence and high model accuracy.
\end{enumerate}

The rest of the paper is organized as follows. We survey related literature in \cref{sec:related_works}. In \cref{sec:preliminaries}, we present the problem setups.  In \cref{sec:system}, we present detailed \name{} algorithms and system framework. We provided convergence analysis for \name{} in \cref{sec:analysis}. \Cref{sec:evaluations} presents the implementation and evaluations of {\name} system.  Finally, \cref{sec:conclusion} concludes the paper.

\vspace{-8pt}
\section{Related Works}\label{sec:related_works}
\paragraph{Coreset Methods for Deep Learning.}
Coreset methods are effective in reducing computational complexity and memory requirements in deep learning. They are based on selecting a representative subset, or coreset, from the original dataset to retain essential information while significantly reducing data size. Coresets have been successfully applied to tasks like image classification \cite{guo2022deepcore, gang2021character, sener2017active}, natural language processing \cite{bilmes2022submodularity,mirzasoleiman2020coresets}, and reinforcement learning \cite{chakraborty2022posterior, huang2021continual, bostani2021strong}. Several approaches for efficient coreset creation include: 
\begin{inparaenum}[1)]
\item \textit{Geometry Based Clustering} \cite{chen2010super, sener2018active, sohler2018strong}, assuming data points in close proximity share similar properties and forming a coreset by removing clustered redundant data points;
\item \textit{Loss Based Sampling} \cite{toneva2018empirical, bachem2015coresets, paul2021deep}, prioritizing training samples based on their contribution to the error or loss reduction during neural network training and selecting the most important samples to form the coreset;
\item \textit{Decision Boundary Methods} \cite{ducoffe2018adversarial, margatina2021active}, focusing on selecting data points near the decision boundary as the coreset, as they carry more informative content for model training; and 
\item \textit{Gradient Matching Solutions} \cite{mirzasoleiman2020coresets, killamsetty2021grad, pooladzandi2022adaptive}, aiming to select a coreset that closely approximates the gradients produced by the full training dataset during deep model training, ensuring minimal gradient differences.
\end{inparaenum}
In this paper, we adopt gradient matching methods to construct distributed coresets across federated learning clients. By utilizing per-sample gradients produced during model training, coresets can be efficiently computed with minimal overhead.

\vspace{-8pt}

\paragraph{Straggler Prevention in Federated Learning.}
Stragglers, slow or unresponsive clients in federated learning, can significantly impact training efficiency and model convergence. Various strategies have been proposed to address this challenge, including:
\begin{inparaenum}[1)]
\item \textit{Client Selection} methods \cite{nishio2019client, lai2021oort, kim2021autofl, nishio2019client} mitigate the impact of stragglers by adaptively selecting a subset of clients based on their performance, training speed, or other criteria. However, this approach may introduce bias in heterogeneous settings, as stragglers with unique and important learning samples could be excluded;
\item \textit{Asynchronous FL} techniques \cite{xie2019asynchronous, nguyen2022federated, huba2022papaya, van2020asynchronous, li2021stragglers, chai2021fedat} eliminate the need for synchronized communication, enabling clients to update local models and communicate with the server independently. Although asynchronous FL can reduce straggler impact, it may suffer from staleness and inconsistency issues affecting the model performance;
\item \textit{Accommodating Partial Work from Stragglers} approaches \cite{li2020federated, li2019feddane, wang2020tackling} adjust local epoch numbers or allow clients to perform partial updates. \fedprox{} \cite{li2020federated} introduces a proximal term in the optimization process, accommodating partial updates without severely affecting model convergence.
\end{inparaenum}
In this paper, we propose \name{}, a novel straggler-resilient training method based on partial-work. Unlike most existing works reducing the number of local epochs, \name{} reduces the number of training samples by creating a coreset. This approach enables \name{} to perform more local optimization steps and explore gradients more deeply, resulting in faster convergence speed and better model accuracy.

\vspace{-8pt}
\section{Preliminaries}\label{sec:preliminaries}

\subsection{Federated Leaning System Setup} 
Consider a set of clients, $U = \{1, 2, \dots, n\}$. For each client $u^i$, we define $V^i = \{1, 2, \dots, m^i\}$ as the index set of its training samples, where $m^i$ represents the size of the training set. The $j$-th data point in the training set of client $u^i$ is denoted as $(x^i_j, y^i_j)$, with $j \in V^i$. $x^i_j$ and $y^i_j$ represent the data and label, respectively. In Federated Learning (FL), the primary objective is to minimize an empirical risk function using the training data from each client. Given a loss function $L$, a machine learning model $f$, and the model parameter space $\mathcal{W}$, the FL problem can be formulated as:
\begin{equation}
\small
\textstyle
w_* = \argmin_{w\in\mathcal{W}}\mathcal{L}(w),
~~\text{where}~~\mathcal{L}(w) := \sum_{i\in U}p^i\mathcal{L}^i(w),
~~ \mathcal{L}^i(w) := \frac{1}{m^i}\sum_{j\in V^i}\mathcal{L}^i_j(w), 
\label{eq:fl_problem}
\end{equation}
Here, $\small \mathcal{L}^i_j(w) := L(f(w, x^i_j), y^i_j)$ represents the empirical loss for each sample $(x^i_j, y^i_j)$, and {\small $p^i = \frac{m^i}{\sum_{i \in U}m^i}$} is the weight proportional to the training set size. However, privacy concerns prevent a central server from directly accessing the clients' data and solving Eq.\eqref{eq:fl_problem}. As an alternative, FL algorithms require each client to solve a local problem, $w^{i,*} = \argmin_{w\in\mathcal{W}}\mathcal{L}^i(w)$, using their data independently. Through iterative communication rounds, the central server aggregates the local models of each client and approximates the solution to Eq.\eqref{eq:fl_problem}.

In FL, clients typically use gradient descent based algorithms like SGD and ADAM for local training. The objective is to provide an unbiased estimate of the full gradient, denoted as {\small $\nabla\mathcal{L}^i(w) = \sum_{j \in V^i}\frac{\partial\mathcal{L}^i_j}{\partial w}$}. SGD optimizers calculate model-gradients based on randomly selected mini-batches of training samples through all the training samples in $V^i$. This constitutes one \emph{epoch} of training. In conventional FL, each client performs SGD for multiple epochs, i.e., $E$ epochs, before sending its gradients to the central server for global model synchronization. This entire process constitutes one FL \emph{round}. The central server then aggregates the received gradients from participating clients and updates the global model. After multiple rounds, i.e., $R$ rounds, of training and synchronization, the global model converges to a satisfactory performance.

The heterogeneity of client training data size and computational capabilities leads to considerable variation in per round training times in Federated Learning. To illustrate, let $c^i$ represent the computational capability of the $i$-th client, which can be inferred from their hardware specifications. Here, $u_i$ takes $1/c^i$ seconds to train one data sample. Hence, the per-round training time is $E\frac{m^i}{c^i}$, where $E$ is the number of epochs per round. Due to the synchronous nature of FL, slower clients can significantly delay the overall training process, resulting in the \emph{straggler problem}.

\subsection{Distributed Coresets for Federated Learning.}
In \name{}, our goal is to address the straggler problem by strategically selecting a small subset $S^i\subseteq V^i$ of the full training set $V^i$ for each $u^i$. This enables the model to be trained only on the subset $S^i$ while still approximately converging to the globally optimal solution (i.e., the model parameters that would be obtained if trained on the entire $V^i$).

Inspired by existing works in gradient-based coreset construction \cite{killamsetty2021grad, mirzasoleiman2020coresets}, the key idea in \name{} is identifying a small subset $S^i$ with the weighted sum of its elements' gradients closely approximating the full gradient over $V^i$. Unlike previous works, our approach generates distributed coresets across all clients $u^i, i\in U$, while still providing global model convergence properties.

To further resolve the straggler problem, we impose a training deadline $\tau$ on every client, ensuring that each $u^i$ can complete one round of training within $\tau$ seconds using the coreset $S^i$. Consequently, $c^i\tau$ represents the maximum number of data samples that can be processed by $u^i$ within a single training round. We specifically formulate the distributed coresets generation problem as follows:
\begin{equation}
\footnotesize
\textstyle
(S^{i,*}, \mathbf{\delta}^{i,*}) = 
\argmin_{S^i \subseteq V^i, \delta^i \in \mathbb{R}^{|S^i|}_+}
\mathcal{E}^i(w, S^i, \delta^i), 
~~\text{s.t.}~~ |S^i| \leq c^i\tau / E,
~~\forall i \in U.
\label{eq:opt_coreset}
\end{equation}
Here {\footnotesize$\mathcal{E}^i(w, S^i, \delta^i):=\norm{\sum_{j\in V^i}\nabla\mathcal{L}^i_j(w) - \sum_{k\in S^i}\delta^i_k\nabla\mathcal{L}^i_k(w)}$} is the 2-normed distance between the full-set gradient and the coreset gradient when the model parameter is $w$. $\delta^i$ is the weight vector of the coreset elements with $\dim(\delta^i)=|S^i|$ .  

Unfortunately, directly solving the aforementioned optimization problem is infeasible due to three main obstacles:
\begin{inparadesc}
\item[a)] Finding the optimal coreset $(S^{i,*}, \mathbf{\delta}^{i,*})$ is an NP-hard task, due to the combinatorial nature of the problem, even when the per-element gradient, $\mathcal{L}^i_j$, can be calculated through SGD training.
\item[b)] Deep machine learning models have high-dimensional model gradients, containing millions of parameters. Solving the above optimization problem with high-dimensional vectors is practically unmanageable.
\item[c)] Eq.\eqref{eq:opt_coreset} needs to be recalculated for every time the model parameter $w$ gets updated, which further intensifies the computational complexity.
\end{inparadesc}
In the following sections, we introduce the design of \name{} and illustrate how it effectively addresses these challenges.

\vspace{-8pt}
\section{\name{} Algorithm and System}\label{sec:system}
\begin{figure}[ht]
    \centering
    \includegraphics[width=\linewidth]{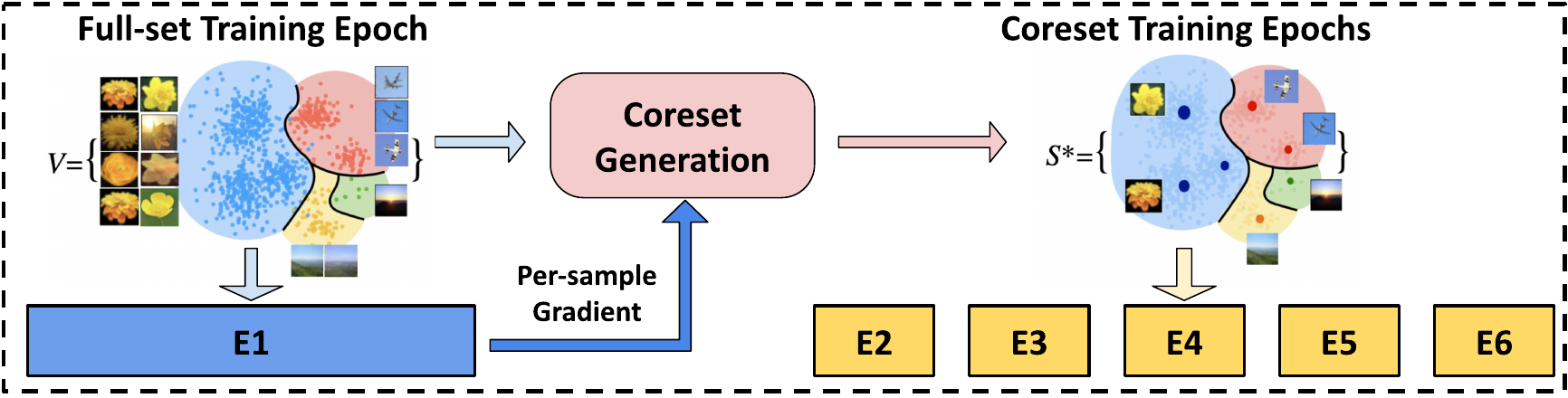}
    \caption{An example workflow of \name{} encompasses a single training round consisting of 6 epochs.}
    \label{fig:system_overview}
\end{figure}

\begin{algorithm}[ht]
\begin{algorithmic}[1]
\footnotesize
    \caption{\name{} Algorithm}
    \label{alg:fedcore_concise}	
    \STATE {\bf Input:}  $K$: {\textit{\# selected clients per round}}. \hspace{4pt}
    $R$: {\textit{\# training rounds}}.  \hspace{4pt} $w_0$: {\textit{\# initial model parameter}}.\\
    {~~~~~~~~~~~} $E$, $\tau$, and $V^i$, $c^i$, $m^i$ $p^i$ for all $i\in U$, as they are defined in \cref{sec:preliminaries}.\\
    \FOR  {$r=0,1, \cdots, R$}
            \STATE Server randomly selects a subset of $K$ clients $U_\text{r}$. Each $u^i$ is chosen with probability $p^i$.\label{line:client_selection} 
            \STATE Server sends current model $w_r$ and round deadline $\tau$ to all chosen clients $u^i, i\in U_\text{r}$.\label{line:send_model} 
            \FOR {\textbf{each} $i \in U_\text{r}$}
                \IF {$E\cdot m^i < c^i\tau$} 
                    \STATE Client $u^i$ executes $E$ epochs of local training with its full-set $V^i$. \label{line:non_straggler}
                \ELSE
                    \STATE Client $u^i$ generates approximated gradient distance, either $\widetilde{d}^i_{j,k}$ or $\widehat{d}^i_{j,k}$ for convex models\\  
                    and neural networks, respectively, over the full-set $V^i$ in the first epoch (\cref{subsec:gradient_approx}). \label{line:straggler_start}
                    \STATE Client $u^i$ constructs coreset $(S^{i,*}, \delta^{i,*})$ by solving the k-medoids problem Eq.\eqref{eq:opt_medoids}.
                    \STATE Client $u^i$ executes $E-1$ epochs of local training with its coreset $(S^{i,*}, \delta^{i,*})$. \label{line:straggler_end}
                \ENDIF
            
            \STATE Client $u^i$ sends its round-end local parameter $w^i_r$ back to the server. \label{line:send_back}
            \ENDFOR
    \STATE Server aggregates the new global model: $w_{r+1} = \frac{1}{K}\sum_{i \in U_{r}} w^i_{r}$. \label{line:aggregate}
    \ENDFOR
\end{algorithmic}
\end{algorithm}

\subsection{\name{} Algorithm  Overview.} 
We present the \name{} workflow in Algorithm \ref{alg:fedcore_concise}. \name{} operates in multiple training rounds, denoted by $R$. Like most existing works \cite{li2020federated}, the server selects $K$ clients randomly, with probabilities proportional to their training set size, i.e., {\small $p^i = \frac{m^i}{\sum_{i \in U}m^i}$} (line \ref{line:client_selection}). The server sends the current model parameter and round deadline $\tau$ to the selected clients for distributed training (line \ref{line:send_model}). Clients assess if they can complete full-set training within $\tau$. If possible, they execute $E$ epochs of SGD training over its full-set $V^i$ (line \ref{line:non_straggler}). Otherwise, they generate a training coreset and train using it (line \ref{line:straggler_start} - \ref{line:straggler_end}). Finally, clients send their local parameters to the server at the end of each training round, which aggregates them to form a new global model (line \ref{line:aggregate}).

To circumvent the need to solve Eq.\eqref{eq:opt_coreset} for every different model parameter $w$ (i.e., every epoch), we design \name{} to search for a suitable coreset periodically at the beginning of each FL round. Figure \ref{fig:system_overview} illustrates the workflow of \name{} during one FL round. In the first epoch, \name{} processes the entire training set, taking a comprehensive initial optimization step and generating per-sample gradients for coreset creation. For the remaining epochs, \name{} operates on a coreset, significantly reducing training time and mitigating the effects of stragglers.

By minimizing the upper bound of gradient estimation dissimilarity (i.e., $\mathcal{E}^i$), we transform Eq.\eqref{eq:opt_coreset} into a \emph{k-medoids problem}, which can be solved approximately in polynomial time (\cref{subsec:kmedoids}). We also use low-dimensional gradient approximations as input for the coreset algorithm instead of high-dimensional model gradients, making coreset generation more efficient and lightweight (\cref{subsec:gradient_approx}). The distributed coresets generated through our approach provide strong global convergence guarantees (\cref{sec:analysis}). In the following sections, we detail the design of these algorithms and discuss practical techniques to accelerate \name{}.

\subsection{Upper Bounding Dissimilarity Estimation with K-Medoids}\label{subsec:kmedoids}
We aim to construct a coreset $S^i$ with $b^i$ elements for each client $u^i$. In order to allow for the first epoch of every training round to be full-set with $m^i$ training elements, we set $b^i = \lfloor \frac{c^i\tau - m^i}{E-1}\rfloor$ to meet the computational capability $c^i\tau - m^i$  of $u^i$ in the remaining $E-1$ epochs. 

To upper bound the dissimilarity between the full-set gradient
and the weighted coreset gradient on $S^i$, first consider a mapping function $\Phi^i: V^i \rightarrow S^i$ that, for every possible model parameter $w \in \mathcal{W}$, assigns each data point $j \in V^i$ to one of its coreset elements $k \in S^i$, i.e., $\Phi^i(j)=k \in S^i$. Let $C^i_k:=\left\{j: \Phi^i(j)=k\right\}\subseteq V^i$ represent the set of data points assigned to data point $k\in S^i$, and let $\delta^i_k:=|C^i_k| \in \mathbb{N}_+$ denote the number of such points. Thus, for any arbitrary $w \in \mathcal{W}$, we have 
{$$\textstyle\sum_{j\in V^i}\nabla\mathcal{L}^i_j(w) - \sum_{k\in S^i}\delta^i_k\nabla\mathcal{L}^i_k(w)= \sum_{j\in V^i}(\nabla\mathcal{L}^i_j(w) - \nabla\mathcal{L}^i_{\Phi^i(j)}(w))$$}
By applying the triangle inequality on both sides, we derive an upper bound for the normed error between the full-set gradient and the weighted coreset gradient, i.e.,
\begin{equation}
\small
\mathcal{E}^i(w, S^i, \delta^i)=
{\textstyle\norm{\sum_{j\in V^i}\nabla\mathcal{L}^i_j(w) - \sum_{k\in S^i}\delta^i_k\nabla\mathcal{L}^i_k(w)}} 
\leq \sum_{j\in V^i}\norm{\nabla\mathcal{L}^i_j(w) - \nabla\mathcal{L}^i_{\Phi^i(j)}(w)}.
\label{eq:upper_bound}
\end{equation}
Note that the upper bound in Eq.\eqref{eq:upper_bound} is minimized when $\Phi^i$ assigns every data point $j \in V^i$ to the element $k \in S^i$ with the most similar gradient, i.e., $\Phi^i(j)=\argmin_{k\in S^i}d^i_{j,k}(w)$, where $d^i_{j,k}(w)=\norm{\nabla\mathcal{L}^i_j(w) - \nabla\mathcal{L}^i_k(w)}$. Hence,
\begin{equation} 
\small
\textstyle
\min_{S^i\subseteq V^i, \delta^i\in \mathbb{N}^{|S^i|}_+}
\mathcal{E}^i(w, S^i, \delta^i) \leq \min_{S^i\subseteq V^i}\left\{
\sum_{j\in V^i}\min_{k\in S^i}d^i_{j,k}(w)\right\}.
\label{eq:upper_bound_min}
\end{equation}
Recall that the minimum value of Eq.\eqref{eq:opt_coreset} is further upper bounded by the left hand side of Eq.\eqref{eq:upper_bound_min}, as it has a larger feasible set for its weight vector $\delta^i\in \mathbb{R}^{|S^i|}_+$.  Hence, we can adjust the optimization objective of Eq.\eqref{eq:opt_coreset} to the right hand side gradient dissimilarity upper bound as follows:
\begin{equation}
\small
\textstyle
(S^{i,*}, \delta^{i,*}) = \argmin_{S^i\subseteq V^i}\left\{\sum_{j\in V^i}\min_{k\in S^i}d^i_{j,k}(w)\right\}, 
~~\text{s.t.}~~ |S^i| \leq b^i,
~~\forall i \in U,
\label{eq:opt_medoids}
\end{equation}
where $\delta^i \in \mathbb{N}^{|S^i|}_+$ is the weight vector associated with $S^i$, given by 
$$\small\delta^{i,*}_k=\left|\left\{j\in V^i:k=\argmin_{l\in S^{i,*}} d^i_{j,l}(w)\right\}\right|.$$
Note that Eq.\eqref{eq:opt_medoids} is a \emph{k-medoids problem} with a budget size of $b^i$. The goal is to minimize the objective function by finding the $b^i$ medoids of the entire training set in the gradient space.

% The \emph{k-medoids problem} is a clustering technique forming $k$ clusters based on data point similarities. K-medoids use actual data points as cluster centers, i.e., the medoids, minimizing dissimilarities between data points and their respective medoids.
% These medoids form a coreset for our Federated Learning (FL) problem. Various algorithms \cite{mirzasoleiman2020coresets, kaufman2009finding, sheng2006genetic} have been developed to address this problem, each offering different computational efficiency and clustering quality trade-offs.
% For our specific problem, we opt for the FasterPAM algorithm \cite{schubert2019faster}, which delivers both speed and accuracy in pinpointing optimal medoids, enabling us to effectively minimize the Eq.\eqref{eq:opt_medoids}. In practice, FasterPAM can quickly solve the k-medoids problem and generate coresets for datasets with thousands of samples within a second.

The \emph{k-medoids problem} is a clustering technique forming $k$ clusters based on data point similarities. K-medoids use actual data points as cluster centers, i.e., the medoids, minimizing dissimilarities between data points and their respective medoids. These medoids form a coreset for our Federated Learning problem. Multiple algorithms \cite{mirzasoleiman2020coresets, kaufman2009finding, sheng2006genetic} have been proposed for this problem, offering diverse computational efficiency and quality trade-offs. In our case, we employ the FasterPAM algorithm, known for its speed and accuracy in identifying optimal medoids, efficiently minimizing our equation Eq.\eqref{eq:opt_medoids}. In essence, FasterPAM quickly solves the k-medoids problem, generating coresets for large datasets within one second.

\subsection{Accelerating Coreset Generation with Gradient Approximation.}\label{subsec:gradient_approx}
Solving the k-medoids problem for each $w\in\mathcal{W}$, as illustrated in Eq.\eqref{eq:opt_medoids}, requires calculating every pairwise gradient difference for the entire training set (i.e., $d^i_{j,k}(w), \forall j,k \in V^i$). Nonetheless, directly computing the gradient-distances is computationally costly due to the typically high-dimensional nature of the full model gradient, especially in the case of deep neural networks with millions of parameters. This leads to a computationally burdensome k-medoids clustering process.
Following the approach in \cite{mirzasoleiman2020coresets}, we tackle this challenge by substituting the full gradient differences with lightweight approximations for two general types of machine learning models as below.

\paragraph{Convex Machine Learning Models.} 

We utilize the method from \cite{allen2017katyusha} that allows for effective gradient distance approximation in convex machine learning models like linear regression, logistic regression, and regularized SVMs. This method approximates the gradient difference between data points using their Euclidean distance, a principle that uniformly applies across the entirety of the parameter space, $\mathcal{W}$. By substituting $d^i_{j,k}(w)$ with $\widetilde{d}^i_{j,k}(w) = \norm{x^i_j - x^i_k}$ in Eq.\eqref{eq:opt_medoids}, the coreset problem is reframed into a 2-norm k-medoids clustering within the original data space. This adjustment facilitates coresets formation using pre-calculated pairwise Euclidean distances, eliminating per-round generation and reducing training-time cost.

\paragraph{Deep Neural Networks.} 
In deep neural networks, gradient changes primarily reflect the loss function's gradient relative to the last layer's input \cite{katharopoulos2018not}. The normed differences of gradients between data points can be effectively bounded as below:
\begin{equation*}
\footnotesize
\textstyle \forall i, j, k,~~~~
d^i_{j,k}(w) \leq \widehat{d}^i_{j,k}(w) = 
c_1 \cdot \norm{
\partial\mathcal{L}^i_j(w) / \partial z^{i}_j-\partial\mathcal{L}^i_k(w) / \partial z^{i}_k}
+ c_2,
\label{eq:nn_bound}
\end{equation*}
Here, $z^i_j$ is the input to the last neural network layer from data point $x^i_j$, and $c_1$ and $c_2$ are constants. We substitute $d^i_{j,k}$ with $\widehat{d}^i_{j,k}$ in Eq.\eqref{eq:opt_medoids} for optimization. $\footnotesize\norm{\partial\mathcal{L}^i_j(w) / \partial z^i_j}$ is attainable from the first epoch of full-set training and requires no extra computation. In \name{}, we derive $\widehat{d}^i_{j,k}$ for all pairs $j, k\in S^i$ in the first FL epoch, thus alleviating the load of high-dimensional k-medoids clustering.

\subsection{Discussions of Design Choices.}\label{subsec:system_discussion}
In designing \name{}, we intentionally set the first epoch to train on the entire dataset, generating (approximated) per-sample gradients for k-medoids coreset generation.
However, heavy loaded straggling clients may struggle to complete the initial epoch\footnote{Existing solutions like \fedprox{} also fail in extreme cases.}, i.e., $c^i\tau < m^i$. 
In such cases, \name{} can use faster coreset methods not requiring a full epoch of forward and backward propagation. As explained in \cref{subsec:gradient_approx}:
\begin{inparadesc}
\item[a)] Convex FL models can use static coresets to achieve model convergence and train with pre-computed coresets in any epoch, i.e., calculate corsets with  pre-computed $\widetilde{d}^i_{j,k}$;
\item[b)] Deep neural networks compute approximated pairwise gradient distance, $\widehat{d}^i_{j,k}$, which is attainable almost as cheap as calculating the loss (with only one step of gradient calculation for the last layer input), instead of a full epoch of forward and backward propagation. 
\end{inparadesc}
As long as the training deadline allows, \name{} prefers to retain the initial full-set epoch, since it offers a more comprehensive representation of the training status by utilizing the entire dataset and establishing a more accurate, well-informed step in beginning of each round of model training.
\vspace{-8pt}
\section{Convergence Analysis}\label{sec:analysis}
The convergence of \name{} is established for strongly convex functions $\Lc$ under mild assumptions. It is important to note that most existing works on the convergence analysis of federated learning (e.g., \cite{haddadpour2019convergence, li2020federated, li2019convergence}) assume that local gradient estimations at the client level are unbiased since the data is directly sampled from the full-set. However, in \name{}, gradients computed from coresets are biased approximations to full-set gradients. As a result, the main technical contribution of our convergence analysis is to meticulously incorporate the coreset gradient approximation error with the federated optimization error. 

\begin{Theorem}\label{thm:main_conv}
    Assume that for any $i\in U$, the loss function $\Lc^i$ is $L$-smooth and $\mu$-strongly convex, and the coreset $\left(S^{i,*}, \delta^{i,*}\right)$ constructed in \name{} is an $\epsilon$-approximation to the full-set, i.e.
    \begin{equation}
        \footnotesize
        \forall w\in\Wc, \quad \frac{1}{m^i}\norm{\sum_{j\in V^i}\nabla\mathcal{L}^i_j(w) - \sum_{k\in S^{i,*}}\delta^{i,*}_k\nabla\mathcal{L}^i_k(w)}\leq\epsilon, 
        % \quad \text{where }\delta^{i,*}=\left(\gamma^i_k\right)_{k\in S^{i,*}}.
    \end{equation}
    \noindent Consider \name{} with $R$ rounds with each round containing $E$ epochs. Set the learning rate 
    $\eta_t=\Omega(1/t)$ for $t\in\{1,2,\cdots,ER\}$.
    The model $w_\text{out}$ output by \name{} after $R$ rounds satisfies
    $$\E\left[\Lc(w_\text{out})-\Lc(w_*)\right] \leq \Oc(\epsilon) + \Oc(1/R),$$
    where $w_*=\argmin_{w\in\Wc}\Lc(w)$ is the global optimum of $\Lc$ in Eq.\eqref{eq:fl_problem}, and the expectation is taken over the randomness in client selection, coreset construction and model initialization.
\end{Theorem}

The comprehensive collections of the technical assumptions and the detailed statement of Theorem \ref{thm:main_conv} can be found in \cref{subapp:assumptions} and \cref{subapp:theorem}. The bound in Theorem \ref{thm:main_conv} indicates that \name{} converges to the global optimum at the rate $\Oc(1/R)$, with an additional cost of $\Oc(\epsilon)$ attributed to the coreset gradient approximation error. It is worth noting that the rate $\Oc(1/R)$ aligns with the existing convergence results for federated learning \cite{haddadpour2019convergence, li2020federated, li2019convergence}.
The trade-off between full-set FL and coreset FL is explicitly characterized in Theorem \ref{thm:main_conv}. While learning on the full-set may circumvent the gradient approximation error, the straggler problem in full-set FL can lead to a small number of training rounds $R$ under a limited time budget. On the other hand, \name{} reduces the impact of the straggler problem and allows for more training rounds to achieve a smaller optimization error $\Oc(1/R)$, while keeping the gradient approximation error low (only $\Oc(\epsilon)$), enabling both efficient and accurate optimization.
The proof of Theorem \ref{thm:main_conv} is deferred to \cref{subapp:theorem}.

\vspace{-8pt}
\section{Evaluations}\label{sec:evaluations}

\subsection{Experimental Setups}\label{subsec:methodology}

\paragraph{FL Datasets and Benchmarks}
\begin{figure}[ht]
\begin{minipage}{0.4\linewidth}
\centering
\captionsetup{type=table} %% tell latex to change to table
\scriptsize
\begin{tabular}{m{1cm}m{0.6cm}m{0.8cm}m{0.5cm}m{0.6cm}}
\hline\hline
\multirow{2}{*}{\textbf{Dataset}} & \multirow{2}{*}{\textbf{Clients}} & \multirow{2}{*}{\textbf{Samples}} & \multicolumn{2}{c}{\textbf{Samples /  Client}} \\ \cline{4-5} 
           &       &         & mean  & std   \\ \hline
MNIST      & 1,000 & 69,035  & 69    & 106   \\
Shakespare & 143   & 517,106 & 3,616 & 6,808 \\
Synthetic  & 30    & 20,101  & 670   & 1,148 \\ \hline\hline
\end{tabular}
\caption{Statistics of the benchmarks}\label{tab:dataset_stats}
\end{minipage}
\hspace{1em}
\begin{minipage}{0.54\linewidth}
\centering
\includegraphics[width=\linewidth]{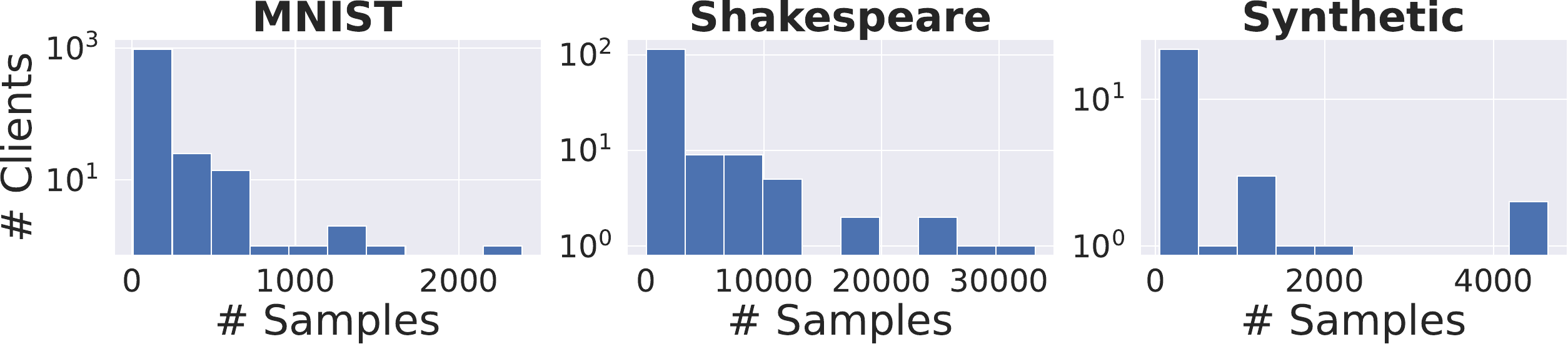}
\caption{Distribution of training samples per client}\label{fig:dataset_stats}
\end{minipage}
\end{figure}
We assess \name{} using three widely recognized federated learning benchmarks from computer vision, natural language processing, and feature-based classification domains. These tasks encompass major machine learning model types, including CNN, RNN, and Logistic Regression (LR), as detailed below:
\begin{inparadesc}
\item[1)] \textbf{MNIST Dataset} \cite{lecun2010mnist}: This dataset features a digit classification task using a three-layer CNN for training. To create statistical heterogeneity, the data is allocated among 1,000 clients, where each client has samples of just two distinct digits. The quantity of samples per client adheres to a power-law distribution, highlighting the diversity among clients.
\item[2)]  \textbf{Shakespeare Dataset} \cite{mcmahan2017communication}: This dataset represents a next-character prediction task trained on \emph{the Complete Works of William Shakespeare} using an LSTM model. Each of the 143 speaking roles in the plays is associated with a distinct client. And  
\item[3)] \textbf{Synthetic Dataset} \cite{li2020federated}: This dataset involves a feature-based classification task with 30 clients training an LR model. Each client's training data is generated from a random function $G(\alpha, \beta)$, where $\alpha$ and $\beta$ control the cross-client and within-client data heterogeneity. Following the approach in \cite{li2020federated}, we evaluate our method with three different parameter settings: $(\alpha, \beta)$ equals to $(0,0)$, $(0.5,0.5)$ and $(1,1)$, respectively.
\end{inparadesc}
In our evaluation, we train MNIST, Shakespare and Synthetic benchmarks for 100, 30 and 100 rounds, respectively. For all three tasks, each round comprises 10 local epochs. Detailed statistics for these three datasets can be found in Table \ref{tab:dataset_stats} and Figure \ref{fig:dataset_stats}.

\paragraph{Comparision Baselines}
We compare \name{} with the following three baselines.
\begin{enumerate}[leftmargin=*, label= {\textbf{\alph*}})]
\item \fedavg{} \cite{mcmahan2017communication} updates the global model by averaging local model updates from participating clients. However, it does not consider training deadlines, and thus, is prone to the stragglers issue.
\item \fedavgds{} \cite{mcmahan2017communication} is a variant of \fedavg{} enforces training deadlines for each round by excluding stragglers. This strategy, however, may negatively impact its overall training performance.
\item \fedprox{} \cite{li2020federated} is designed to handle partial results from stragglers that might complete fewer local epochs than anticipated, \fedprox{} incorporates a quadratic proximal term that explicitly limits the magnitude of local model updates to accommodate stragglers.
\end{enumerate}

\paragraph{Implementations}
We develop \name{} along with all the baseline algorithms using PyTorch \cite{paszke2019pytorch}, extending the simulation framework proposed in FedML\cite{he2020fedml}. For each client $u^i$, we sample its computational capability from a normal distribution, i.e., $c^i\sim \mathcal{N}(1,\,0.25)$. As discussed in \cref{sec:preliminaries}, the per-round training time for a client is proportional to $\frac{m^i}{c^i}$. To emulate the stragglers problem, we designate the slowest $s\%$ of clients as stragglers by setting a per-round training deadline that these clients cannot complete all their training tasks within the allotted time. When the training deadline is reached, \fedavgds{} simply excludes all stragglers and aggregates a global model using the non-stragglers' gradients. In contrast, \fedprox{} and \name{} employ different strategies such as reducing local training epochs or training with coresets. In our evaluation, we consider two different stragglers' settings by choosing $s$ to be 10 and 30, respectively. A more detailed implementation and hyper-parameters for the evaluation are presented in \cref{app:evaluation}.

\subsection{Evaluation Results}\label{subsec:results}

\begin{figure}[ht]
    \centering
    \includegraphics[width=\linewidth]{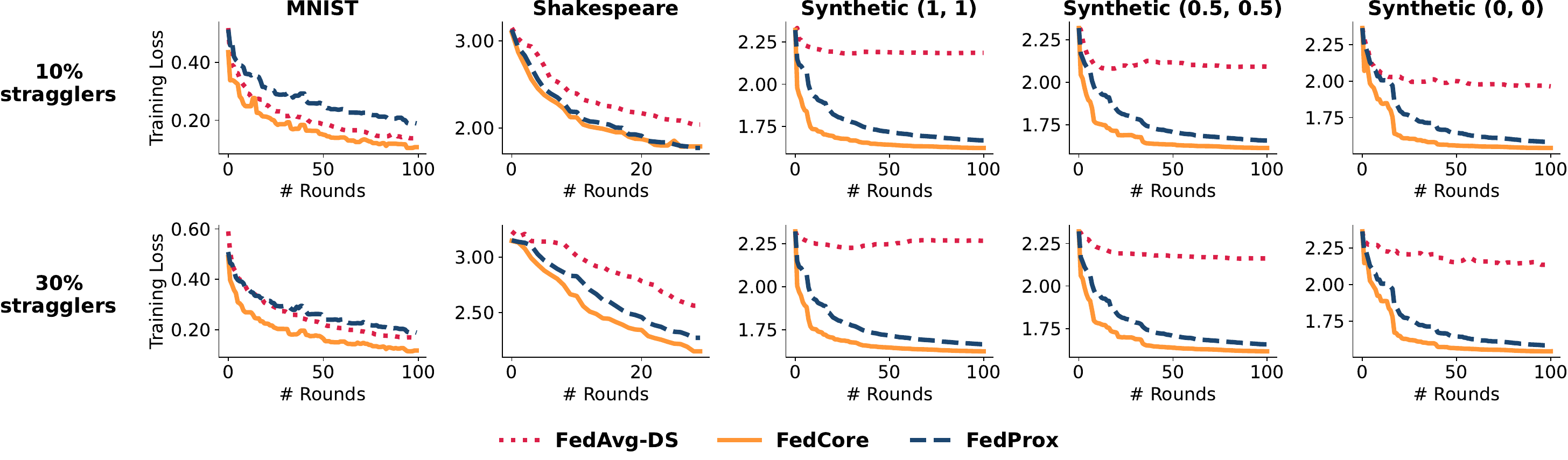}
    \caption{The training loss curves for \fedavgds{}, \name{}, and \fedprox{} at 10\% and 30\% stragglers.}
    \label{fig:training_loss}
\end{figure}

\begin{table}[ht]
\centering
\scriptsize
\begin{tabular}{cccccccccccc}
\hline\hline
 &  & \multicolumn{2}{c}{\textbf{MNIST}} & \multicolumn{2}{c}{\textbf{Shakespeare}} & \multicolumn{2}{c}{\textbf{Synthetic (1, 1)}} & \multicolumn{2}{c}{\textbf{Synthetic (0.5, 0.5)}} & \multicolumn{2}{c}{\textbf{Synthetic (0, 0)}} \\ \cline{3-12} 
\textbf{} &  & 10\% & 30\% & 10\% & 30\% & 10\% & 30\% & 10\% & 30\% & 10\% & 30\% \\ \hline\hline
 & FedAvg & \multicolumn{2}{c}{94.7} & \multicolumn{2}{c}{44.9} & \multicolumn{2}{c}{71.8} & \multicolumn{2}{c}{73.7} & \multicolumn{2}{c}{88.2} \\ \cline{3-12} 
 & FedAvg-DS & 94.1 & 93.1 & 39.0 & 25.2 & 23.0 & 19.9 & 32.2 & 23.6 & 36.3 & 34.6 \\
 & FedProx & 92.6 & 92.7 & 44.1 & 31.3 & \textbf{72.3} & 72.2 & 74.1 & 74.1 & 87.2 & 87.2 \\
\multirow{-4}{*}{\begin{tabular}[c]{@{}c@{}}Test \\ Accuracy\end{tabular}} & \textbf{FedCore} & \textbf{94.6} & \textbf{94.5} & \textbf{44.7} & \textbf{34.8} & 72.2 & \textbf{72.8} & \textbf{75.2} & \textbf{75.1} & \textbf{88.5} & \textbf{88.3} \\ \hline\hline
 & FedAvg & {\color{red} \textbf{3.27}} & {\color{red} \textbf{8.48}} & {\color{red} \textbf{1.38}} & {\color{red} \textbf{4.09}} & {\color{red} \textbf{1.37}} & {\color{red} \textbf{4.80}} & {\color{red} \textbf{1.37}} & {\color{red} \textbf{4.80}} & {\color{red} \textbf{1.37}} & {\color{red} \textbf{4.80}} \\ \cline{3-12} 
 & FedAvg-DS & 0.94 & 0.95 & 0.60 & 0.67 & 0.69 & 0.79 & 0.69 & 0.79 & 0.69 & 0.79 \\
 & FedProx & 0.98 & 0.99 & 0.85 & 0.94 & 0.86 & 0.95 & 0.86 & 0.95 & 0.86 & 0.95 \\
\multirow{-4}{*}{\begin{tabular}[c]{@{}c@{}}Mean Training\\ Time per Round\\ (normalized)\end{tabular}} & \textbf{FedCore} & 0.99 & 0.99 & 0.90 & 0.99 & 0.93 & 0.99 & 0.93 & 0.99 & 0.93 & 0.99 \\ \hline\hline
\end{tabular}
\vspace{3pt}
\caption{Comparison of test accuracy and training time for \fedavg{}, \fedavgds{}, \fedprox{}, and \name{} at 10\% and 30\% stragglers. \textbf{Bold}: top accuracy; {\color{red} \textbf{Red}}: exceeded deadline. Normalized time of 1 is round deadline.}
\vspace{-10pt}
\label{tab:acc-time}
\end{table}

\paragraph{Model Performance}
We present the training loss curves in Figure \ref{fig:training_loss} and model accuracy, along with normalized training time, in Table \ref{tab:acc-time}. More evaluation results are depicted in \cref{app:evaluation}.
For model training loss, \name{} consistently achieves the fastest convergence speed and yields the lowest model loss. In contrast, \fedavgds{} struggles to converge well under synthetic benchmarks due to its approach of dropping stragglers, which contain unique training samples essential for learning. \fedprox{} presents competitive performance, but with slower convergence and higher loss compared to \name{}.
Concerning test accuracy, \name{} consistently achieves the highest or near-highest values across all datasets and stragglers' settings, highlighting its superior performance in maintaining or improving model accuracy even with stragglers. \fedprox{} also demonstrates competitive performance, owing to its ability to accommodate partial results from stragglers, which contain a significant amount of unique training samples that improve model accuracy. However, \fedavgds{} often results in lower accuracy, particularly in the 30\% stragglers setting, as its approach of dropping straggler clients negatively impacts training performance.
In terms of training time, \name{}, \fedprox{}, and \fedavgds{} are deadline-aware, ensuring they do not exceed the round deadlines. While \name{} does not always achieve the fastest training time, it strikes a balance between efficiency and maintaining high accuracy. Conversely, \fedavg{} exhibits the longest training times, indicated in red, showcasing its vulnerability to stragglers and lack of deadline-awareness.

\begin{figure}[ht]
 \centering
 \begin{minipage}[b]{0.46\textwidth}
 \centering
 \includegraphics[width=\textwidth]{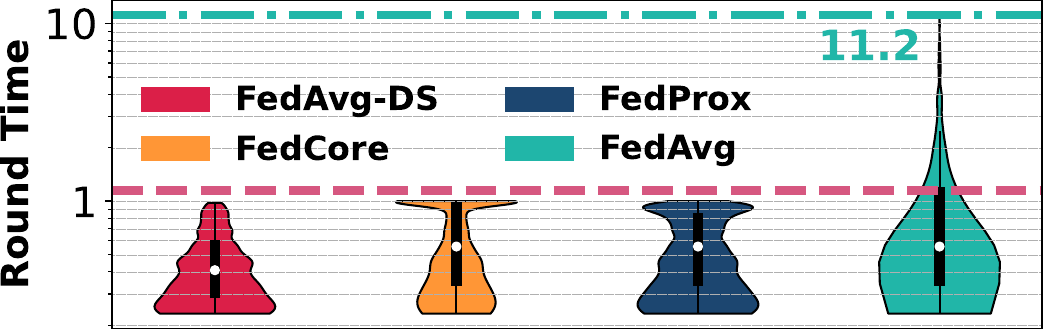}
 \caption{Round length distribution on MNIST benchmark, 30\% stragglers. The y-axis is presented in log-scale for better illustration.}
 \label{fig:violin_minist}
 \end{minipage}
 \hfill
 \begin{minipage}[b]{0.48\textwidth}
 \centering
 \includegraphics[width=\textwidth]{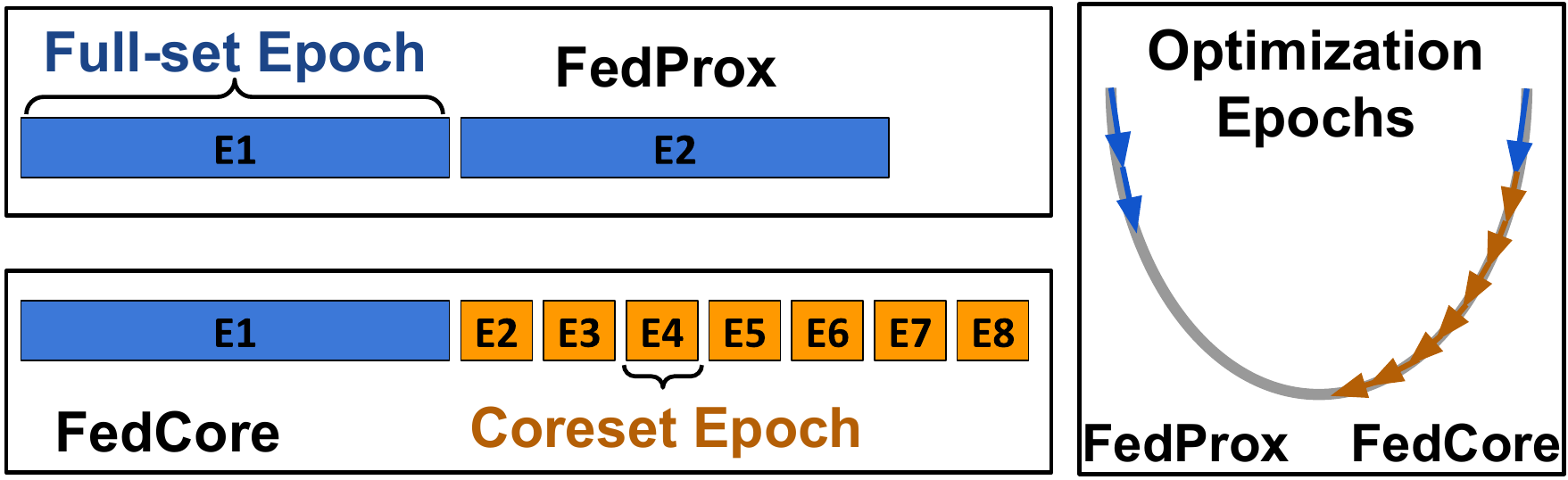}
 \caption{Faster \name{} convergence vs. \fedprox{}, due to more coreset-based gradient steps compared to \fedprox{}'s fewer epochs of full-set training.}
 \label{fig:core_vs_prox}
 \end{minipage}
\end{figure}

\paragraph{Stragglers Handling}

Figure  \ref{fig:violin_minist} presents the distribution of clients' round times for the  MNIST benchmark with 30\% stragglers. As the figure illustrates, \fedavg{}, which is oblivious to round deadlines, generates a tail distribution that can exceed 11 times the allotted training time for a round. In contrast, deadline-aware algorithms like \name{}, \fedavgds{}, and \fedprox{} consistently ensure that each training round is completed before the deadline. Interestingly, the \name{} distribution is more tightly clustered around the round deadline in comparison to \fedavgds{} and \fedprox{}, which signifies a more effective utilization of the allotted training time to accurately follow the gradient direction. Table \ref{tab:acc-time} shows that although \name{} requires slightly longer time than the other two deadline-aware algorithms, it successfully meets the deadline requirements and ultimately achieves the best model performance. 

As depicted in Figure \ref{fig:core_vs_prox}, \name{} takes advantage of coresets to perform more epochs of local optimization and deeper gradient exploration, as opposed to \fedprox{}'s fewer epochs of full-set training. This approach leads to a faster convergence rate and improved model accuracy, demonstrating the effectiveness of the \name{} algorithm in addressing the straggler problem in federated learning.
% \begin{wrapfigure}{r}{0.5\textwidth} 
%  \centering
%  \includegraphics[width=\linewidth,trim=5pt 0pt 0pt 65pt]{figures/core_vs_prox.pdf}
% \caption{Illustration: Faster \name{} convergence vs. \fedprox{}, due to more coreset-based gradient steps compared to fewer epochs of full-set training.}
%  \label{fig:core_vs_prox}
% \end{wrapfigure}

\vspace{-8pt}
\section{Concluding Remarks}\label{sec:conclusion}

In this paper, we introduce \name{}, an innovative algorithm addressing the straggler problem in federated learning using distributed coresets. \name{} effectively adapts to updated models and integrates coreset generation with minimal overhead, significantly outperforming traditional methods. Our comprehensive analysis and evaluation demonstrate that \name{} substantially reduces FL training time while maintaining high accuracy. With regards to \textbf{broader impacts}, this research pioneers the use of coreset methods in efficient federated learning, paving the way for more scalable and robust systems, especially in privacy-sensitive domains where data protection is vital.

\def\bibfont{\footnotesize}
\bibliographystyle{plainnat}
\bibliography{body/reference.bib}

\newpage
\appendix
\section{Convergence Analysis Proof}\label{app:analysis}

\subsection{Problem Settings and Notations}\label{subapp:notations}
\paragraph{Problem Set-up}
First recall the notations defined in section \ref{sec:preliminaries}. The federated learning problem is to solve
\begin{equation}
w_* = \argmin_{w\in\mathcal{W}}\mathcal{L}(w),
~~\text{where}~~\mathcal{L}(w) := \sum_{i\in U}p^i\mathcal{L}^i(w),
~~ \mathcal{L}^i(w) := \frac{1}{m^i}\sum_{j\in V^i}\mathcal{L}^i_j(w).
\label{eq:fl_problem_app}
\end{equation}
with $\mathcal{L}^i_j(w) := L(f(w, x^i_j), y^i_j)$ representing the empirical loss for each sample $(x^i_j, y^i_j)$ from the $i$-th client, under the model $f(w,\cdot)$. Here $|U|=n$ is the total number of clients, and $p^i$ is the weight of the $i$-th client, proportional to the size $m^i$ of its training set with $\sum_{i=1}^n p^i=1$.

The proposed federated learning algorithm \name{} consists of $R$ rounds, each of which contains $E$ epochs. We use the time index $t\in\{0,1,2,\cdots, E R\}$ to denote the time step of each epoch, where $t=0$ corresponds to the model initialization. Meanwhile, denote $w_t^i$ to be the model parameter of client $i$ at time step $t$. One typical round in \name{} is described as follows.

Let $t=(r-1)E$ be the beginning at the $r$-th round for some $r=1,2,\cdots,R$. The central server broadcasts the latest model, $w_{t}$, to all the devices:
$$w_t^i\longleftarrow w_t,\quad \forall i\in U.$$
After that, the central server selects a set $U_t$ of $K$ clients randomly from $U$, according to the sampling probabilities $p^i, i\in U$. The coreset is then constructed for each client $(S^{i,*}, \delta^{i,*}), i\in U_t$. Each client $i\in U_t$ performs local updates on its model $w_t^i$ for the remaining epochs in the current round, based on the data in its coreset $(S^{i,*}, \delta^{i,*})$:
\begin{equation}\label{eqn:one_epoch_update}
    w_{t+k+1}^{i} \longleftarrow w_{t+k}^{i}-\eta_{t+k} g_{t+k}^i,\text{ for } k=0,1, \cdots, E-1,
\end{equation}
where $\eta_{t+k}$ is the learning rate and $g_{t+k}^i$ is the gradient computed from $(S^{i,*}, \delta^{i,*})$:
\begin{equation}\label{eqn:grad_core}
g_{t+k}^i=\frac{1}{m^i}\sum_{j\in S^{i,*}}\delta^{i,*}_j\nabla\mathcal{L}^i_j(w_{t+k}^i).
\end{equation}
Finally, at the end of the $r$-th round, the server aggregates the local models $\left\{w_{t+E}^i\right\}_{i\in U_t}$ to produce the new global model $w_{t+E}$:
\begin{equation}\label{eqn:agg}
    w_{t+E}\longleftarrow \frac{1}{K}\sum_{i\in U_t}w_{t+E}^i.
\end{equation}

Note that the current update in Eq.\eqref{eqn:one_epoch_update} is written in the form of gradient descent (GD), meaning that the model will be updated once based on the full gradient computed from $(S^{i,*}, \delta^{i,*})$. In practice, however, within one epoch, the update in Eq.\eqref{eqn:one_epoch_update} is usually conducted sequentially using stochastic gradient descent (SGD): the entire coreset will be randomly split into several mini-batches, and the parameter will be updated on each mini-batch. In the following analysis, we focus on the gradient descent setting in Eq.\eqref{eqn:one_epoch_update} 
for the ease of presentation. Convergence guarantees for SGD updates can be established by using the similar arguments as in the proofs of our main results.

\paragraph{Notations} In subsequent analysis, we use 
\begin{equation}\label{eqn:G_t_i}
    G_t^i:=\nabla\mathcal{L}^i(w_t^i)=\frac{1}{m^i}\sum_{j\in V^i}\nabla\mathcal{L}^i_j(w_t^i)
\end{equation}
to denote the full gradient from the full-set $V^i$ of client $i$ at time $t$. And denote 
\begin{equation}\label{eqn:G_t}
    G_t=\sum_{i\in U}p^iG_t^i=\sum_{i\in U}p^i\nabla\mathcal{L}^i(w_t^i)
\end{equation}
as the full gradient of the population at time $t$. Meanwhile, denote 
\begin{equation}\label{eqn:g_t_i}
    g_{t}^i=\frac{1}{m^i}\sum_{j\in S^{i,*}}\delta^{i,*}_j\nabla\mathcal{L}^i_j(w_{t}^i),
\end{equation}
where $\left(\delta^{i,*}, S^{i,*}\right)$ is defined in Eq.\eqref{eq:opt_medoids}, as the gradient computed from the coreset of client $i$ at time $t$. And denote $g_{t}=\sum_{i\in U}p^i g_t^i$ as the population coreset gradient at time $t$.

In practice, within one round, only a subset of $K$ randomly selected clients will update their parameters, and the choices of clients vary each round. In order to facilitate the analysis under the random selection scheme, the following
thought trick is introduced to circumvent the difficulty: we assume that \name{} always activates \textbf{all} devices at the beginning of each round, while only aggregates parameters from those sampled devices an the end of one round. It is clear that this updating scheme is equivalent to the original. More specifically, the updating scheme in \name{} is given by, $\forall i \in U$,
\begin{equation}\label{eqn:v_w_update}
\begin{aligned}
v_{t+1}^{i} & = w_{t}^{i}-\eta_{t} g_t^i, \\
w_{t+1}^{i} & = \begin{cases}v_{t+1}^{i} & \text { if } t+1 \notin \mathcal{I}_{E}, \\
\frac{1}{K}\sum_{k \in U_t}v_{t+1}^{k} & \text { if } t+1 \in \mathcal{I}_{E} ,\end{cases}
\end{aligned}
\end{equation}
where $\mathcal{I}_{E}=\{r E \mid r=1,2,\cdots, R\}$ is the set of global synchronization steps, and $U_t$ is the set of $K$ selected clients at time $t$. An additional variable $v_{t+1}^i$ is introduced to represent the immediate result of one step GD update from $w_t^i$, and $w_t^i$ is the final model parameters maintained by client $i$ at time $t$, (possibly after the global synchronization). 

In addition, two virtual sequences are introduced in the subsequent analysis to denote the population-averaged model parameters, following the ideas from \cite{haddadpour2019convergence, li2019convergence, stich2019local}: 
\begin{equation}\label{eqn:bar_v_w}
    \overline{v}_{t}=\sum_{i\in U} p^{i} v_{t}^{i},\text{ and } \overline{w}_{t}=\sum_{i\in U} p^{i} w_{t}^{i},
\end{equation}
where $\overline{v}_{t+1}$ results from an single GD step of from $\overline{w}_{t}$:
\begin{equation}\label{eqn:w_from_v}
    \overline{v}_{t+1}=\overline{w}_{t}-\eta_{t} g_{t}.
\end{equation}

\subsection{Assumptions and Convergence Results}\label{subapp:assumptions}
The following are the detailed assumptions required for the convergence analysis.

\begin{Assumption}[$L$-smoothness]\label{ass:l_smooth}
    $\forall i\in U, \Lc^i$ is $L$-smooth: for all $v, w\in\Wc$, 
    $$\Lc^i(v) \leq \Lc^i(w)+(v-w)^\top \nabla \Lc^i(w)+\frac{L}{2}\|v-w\|_{2}^{2}.$$
\end{Assumption}

\begin{Assumption}[$\mu$-strong convexity]\label{ass:mu_convex}
    $\forall i\in U, L^i$ is $\mu$-strongly convex: for all $v, w\in\Wc$,
    $$\Lc^i(v) \geq \Lc^i(w)+(v-w)^\top \nabla \Lc^i(w)+\frac{\mu}{2}\|v-w\|_{2}^{2}.$$
\end{Assumption}

\begin{Assumption}[$\epsilon$-coreset]\label{ass:coreset}
    For any client $i$ and time step $t$, with probability one, the coreset gradient $g_t^i$ in Eq.\eqref{eqn:g_t_i} is an $\epsilon$-approximation to the full-set gradient $G_t^i$ in Eq.\eqref{eqn:G_t_i}:
    $$\left\|g_t^i - G_t^i\right\|\leq\epsilon, \quad \forall i\in U, \text{ and } t\in\{0,1,\cdots,ER\}, \quad\text{with probability one}.$$
\end{Assumption}

\begin{Assumption}[$D$-bounded gradient]\label{ass:bounded}
    For any client $i$ and time step $t$, with probability one, 2-norms of the coreset gradient $g_t^i$ in Eq.\eqref{eqn:g_t_i} and the full-set gradient $G_t^i$ in Eq.\eqref{eqn:G_t_i} are uniformly upper bounded by a constant $D>0$:
    $$\max\left\{\left\|g_t^i\right\|, \left\|G_t^i\right\|\right\}\leq D, \quad \forall i\in U, \text{ and } t\in\{0,1,\cdots,ER\}, \quad\text{with probability one}.$$
\end{Assumption}

\begin{Assumption}[$\Gamma$-heterogeneity]\label{ass:heter}
    Let $\Lc_{*}$ and $\Lc^{i}_{*}$ be the minimum values of $\Lc$ and $\Lc^i$, respectively. Assume there is a positive constant $\Gamma>0$ such that $\Gamma\geq\Lc_{*}-\sum_{i\in U} p^{i} \Lc^{i}_{*}$.
\end{Assumption}

\begin{Assumption}[Random sampling]\label{ass:sample}
    For any time step $t$, assume $U_{t}$ contains a subset of $K$ indices randomly selected with replacement according to the sampling probabilities $\left\{p^{i}\right\}_{i\in U}$. 
\end{Assumption}

\paragraph{Comments on Assumptions} Assumption \ref{ass:l_smooth} and \ref{ass:mu_convex} are standard assumptions in convex optimization \cite{boyd2004convex}; typical examples are linear/ridge regression, logistic regression, and regularized support vector machines. 
Assumption \ref{ass:coreset} characterizes the approximation capability of the coreset to the full-set, which is standard in the theoretical works on coreset-based gradient descent methods \cite{mirzasoleiman2020coresets, pooladzandi2022adaptive}. 
Assumption \ref{ass:bounded} on the bounded gradient is a widely adopted setting in the existing theoretical works for federated learning and coreset methods \cite{li2020federated, li2019convergence,  mirzasoleiman2020coresets}. 
Meanwhile, note that Assumptions \ref{ass:coreset} and \ref{ass:bounded} are presented in a probabilistic form to account for the potential randomness resulting from the coreset construction steps in \name{}.
Assumption \ref{ass:heter} quantifies the degree of heterogeneity among different clients. In the special case when data from all the clients are i.i.d., then $\Lc_{*}-\sum_{i\in U} p^{i} \Lc^{i}_{*}\to 0$ as the number of samples grows. 
Assumption \ref{ass:sample} assumes the $K$ clients are selected from the distribution $\left\{p^{i}\right\}_{i\in U}$ independently and with replacement, which is a common set-up in both theoretical and empirical works \cite{li2020federated, li2019convergence}. 

\paragraph{Randomness in \name{}} Note that randomness in \name{} can be attributed to three sources: client selection, coreset construction and model initialization $w_0$. Throughout the subsequent analysis and statements, unless otherwise specified, the expectation $\E[\cdot]$ is be taken over all three sources of randomness. Meanwhile, the notation $\E_{U_t}[\cdot]$ is also introduced to denote the expectation over the random client selection at time $t$, conditioned on the other sources of randomness.

\subsection{Proofs of Main Results}\label{subapp:theorem}
The convergence of \name{} is established by the following theorem, which can be considered as a more detailed version of Theorem \ref{thm:main_conv}.

\begin{Theorem}\label{thm:conv}
    Assume Assumptions \ref{ass:l_smooth}, \ref{ass:mu_convex}, \ref{ass:coreset}, \ref{ass:bounded}, \ref{ass:heter}, \ref{ass:sample} hold with constants $L, \mu, \epsilon, D, \Gamma$. Consider \name{} with $R$ rounds and each round contains $E$ epochs. For $t\in\{0,1,\cdots,ER\}$, set the learning rate 
    $$\eta_t=\frac{\alpha}{t+\beta},\quad\text{with } \alpha=\frac{2}{\mu}\text{ and } \beta=\max\left\{E,\frac{8L}{\mu}\right\}.$$ 
    The model $w_\text{out}$ output by \name{} after $R$ rounds of training satisfies
    \begin{equation}\label{eqn:conv}
        \E\left[\|w_\text{out}-w_*\|^2\right] \leq A_1 + \frac{A_2}{ER+\beta},
    \end{equation}
    where the constants $A_1$ and $A_2$ are given by:
    \begin{equation}\label{eqn:A1_A2}
        \begin{aligned}
        A_1&=\frac{2\epsilon D}{\mu^2},\\
        A_2&=\max\left\{\beta\E\left[\|w_0-w_*\|^2\right], \frac{4}{\mu^2}\left[\frac{4E^2D^2}{K}+8(E-1)^2D^2+6L\Gamma+\epsilon^2+2\epsilon D\right]\right\}.
    \end{aligned}
    \end{equation}
    Here $w_*=\argmin_{w\in\Wc}\Lc(w)$ as defined in Eq.\eqref{eq:fl_problem_app} and the expectation is taken over the randomness in client selection, coreset construction and model initialization $w_0$. Consequently,
    \begin{equation}\label{eqn:conv_value}
        \E\left[\Lc(w_\text{out})-\Lc(w_*)\right] \leq \frac{L}{2}\left(A_1 + \frac{A_2}{ER+\beta}\right).
    \end{equation}
    
\end{Theorem}

The proof of Theorem \ref{thm:conv} is based on the following three key lemmas, whose proofs are deferred to \cref{subapp:theorem}.

\begin{Lemma}\label{lemma:H2}
    Under the setting of Theorem \ref{thm:conv}, for $t+1 \in \mathcal{I}_{E}=\{r E \mid r=1,2,\cdots, R\}$, the set of global synchronization steps,
    \begin{equation}\label{eqn:H2}
    \mathbb{E}_{U_{t}}\left[\overline{{w}}_{t+1}\right]=\overline{{v}}_{t+1}.
    \end{equation}
\end{Lemma}

\begin{Lemma}\label{lemma:H1}
    Under the setting of Theorem \ref{thm:conv}, for $t+1 \in \mathcal{I}_{E}=\{r E \mid r=1,2,\cdots, R\}$, the set of global synchronization steps, the expected difference between $\overline{{v}}_{t+1}$ and $\overline{{w}}_{t+1}$ is bounded by
    \begin{equation}\label{eqn:H1}
    \mathbb{E}_{U_{t}}\left[\left\|\overline{{v}}_{t+1}-\overline{{w}}_{t+1}\right\|^{2}\right] \leq \frac{4}{K} \eta_{t}^{2} E^{2} D^{2}.
    \end{equation}
\end{Lemma}

\begin{Lemma}\label{lemma:H3}
    Under the setting of Theorem \ref{thm:conv}, for any time step $t+1 \in \{1,2,\cdots, ER\}$,
    \begin{equation}\label{eqn:H3}
        \mathbb{E}\left[\left\|\overline{{v}}_{t+1}-{w}_*\right\|^{2}\right] \leq\left(1-\eta_{t} \mu\right) \mathbb{E}\left[\left\|\overline{{w}}_{t}-{w}_*\right\|^{2}\right] + \eta_{t}\cdot A_3 + \eta_{t}^{2}\cdot A_4,
    \end{equation}
    where
    \begin{equation}\label{eqn:A3_A4}
        \begin{aligned}
        A_3=\frac{2\epsilon D}{\mu}, \quad A_4=8(E-1)^2D^2+6L\Gamma+\epsilon^2+2\epsilon D.
    \end{aligned}
    \end{equation}
\end{Lemma}

\begin{proof}[Proof of Theorem \ref{thm:conv}]
First note the following decomposition:
\begin{align}\label{eqn:main_decomp}
\left\|\overline{{w}}_{t+1}-{w}_{*}\right\|^{2} & =\left\|\overline{{w}}_{t+1}-\overline{{v}}_{t+1}+\overline{{v}}_{t+1}-{w}_{*}\right\|^{2} \notag\\
& =\underbrace{\left\|\overline{{w}}_{t+1}-\overline{{v}}_{t+1}\right\|^{2}}_{H_{1}} + \underbrace{2\left\langle\overline{{w}}_{t+1}-\overline{{v}}_{t+1}, \overline{{v}}_{t+1}-{w}_{*}\right\rangle}_{H_{2}} + \underbrace{\left\|\overline{{v}}_{t+1}-{w}_{*}\right\|^{2}}_{H_{3}}.
\end{align}
For the first term $H_1$ in Eq.\eqref{eqn:main_decomp}, note that when $t+1 \notin \mathcal{I}_{E}$, we have $\overline{{v}}_{t+1}=\overline{{w}}_{t+1}$, and $H_1$ vanishes. Additionally, if $t+1 \in \mathcal{I}_{E}$, then the expectation of $H_1$ is bounded by Lemma \ref{lemma:H1}: 

For the second term $H_2$ in Eq.\eqref{eqn:main_decomp}, when $t+1 \notin \mathcal{I}_{E}$, $H_2$ vanishes since $\overline{{v}}_{t+1}=\overline{{w}}_{t+1}$. Additionally, when $t+1 \in \mathcal{I}_{E}$, $H_{2}$ vanishes under the expectation $\E_{U_t}[\cdot]$, due to the unbiasedness of $\overline{w}_{t+1}$ stated in Lemma \ref{lemma:H2}.

For the third term $H_3$ in Eq.\eqref{eqn:main_decomp}, its expectation is bounded by Lemma \ref{lemma:H3} for any time step $t+1 \in \{1,2,\cdots, ER\}$.

Overall, combining the bounds on $H_1$, $H_2$ and $H_3$ together, we have for any $t+1 \in \{1,2,\cdots, ER\}$,
\begin{equation}\label{eqn:recur}
    \mathbb{E}\left[\left\|\overline{{w}}_{t+1}-{w}_*\right\|^{2}\right] \leq\left(1-\eta_{t} \mu\right) \mathbb{E}\left[\left\|\overline{{w}}_{t}-{w}_*\right\|^{2}\right] + \eta_{t}\cdot A_3 + \eta_{t}^{2}\cdot \left(\frac{4E^2D^2}{K} + A_4\right),
\end{equation}
where $A_3, A_4$ are defined in Eq.\eqref{eqn:A3_A4}. For simplicity, denote 
\begin{equation}\label{eqn:A_5}
    A_5;=\frac{4E^2D^2}{K} + A_4.
\end{equation}

Now we will prove by induction that under the diminishing step size $\eta_t=\frac{\alpha}{t+\beta}$ with $\alpha=\frac{2}{\mu}$ and $\beta=\max\left\{E,\frac{8L}{\mu}\right\}$, for any time step $t \in \{0,1,\cdots,ER\}$,
\begin{equation}\label{eqn:induction}
    \E\left[\|\overline{{w}}_{t}-w_*\|^2\right] \leq A_1 + \frac{A_2}{t+\beta},
\end{equation}
where $A_1, A_2$ are defined in Eq.\eqref{eqn:A1_A2}. 

First, note that the definition of $A_2$ in Eq.\eqref{eqn:A1_A2} ensures that Eq.\eqref{eqn:induction} holds for $t=0$. Assume Eq.\eqref{eqn:induction} holds for some time step $t$. Then for time step $t+1$, by Eq.\eqref{eqn:recur}, we have
\begin{align}\label{eqn:induction_derive_1}
    \mathbb{E}\left[\left\|\overline{{w}}_{t+1}-{w}_*\right\|^{2}\right] &\leq \left(1 - \frac{\alpha\mu}{t+\beta}\right) \cdot \left(A_1 + \frac{A_2}{t+\beta}\right) + \frac{\alpha}{t+\beta} \cdot A_3 + \left(\frac{\alpha}{t+\beta}\right)^{2} \cdot A_5,\notag\\
    &=A_1 + \left(1 - \frac{\alpha\mu}{t+\beta}\right) \cdot \frac{A_2}{t+\beta} + \left(\frac{\alpha}{t+\beta}\right)^{2} \cdot A_5 + \frac{\alpha(A_3 - \mu A_1)}{t+\beta}
\end{align}
Note that by the definitions of $A_1$ in Eq.\eqref{eqn:A1_A2} and $A_3$ in Eq.\eqref{eqn:A3_A4}, 
\begin{equation}\label{eqn:A_3=A_1}
    A_3 = \mu A_1.
\end{equation}
Meanwhile, 
\begin{align}\label{eqn:induction_derive_2}
    \left(1 - \frac{\alpha\mu}{t+\beta}\right) \cdot \frac{A_2}{t+\beta} + \left(\frac{\alpha}{t+\beta}\right)^{2} \cdot A_5 &= \frac{(t+\beta-1)A_2}{(t+\beta)^{2}} + \left[\frac{\alpha^{2} A_5}{(t+\beta)^{2}}-\frac{(\alpha \mu-1)A_2}{(t+\beta)^{2}}\right]\notag\\
    &\leq\frac{A_2}{t+\beta+1} + \left[\frac{\alpha^{2} A_5}{(t+\beta)^{2}}-\frac{(\alpha \mu-1)A_2}{(t+\beta)^{2}}\right]\notag\\
    &=\frac{A_2}{t+\beta+1} + \frac{1}{(t+\beta)^{2}}\left[\frac{4A_5}{\mu^2}-A_2\right]\notag\\
    &\leq\frac{A_2}{t+\beta+1}.
\end{align}
Here the second equality in Eq.\eqref{eqn:induction_derive_2} is due to the fact that $\alpha=\frac{2}{\mu}$, and the second inequality in Eq.\eqref{eqn:induction_derive_2} comes from the fact that $A_2\geq\frac{4A_5}{\mu^2}$, which is a direct consequence of the definitions of $A_2$ in Eq.\eqref{eqn:A1_A2} and $A_5$ in Eq.\eqref{eqn:A_5}.

Plugging Eq.\eqref{eqn:A_3=A_1} and Eq.\eqref{eqn:induction_derive_2} into Eq.\eqref{eqn:induction_derive_1} completes the proof of the induction hypothesis in Eq.\eqref{eqn:induction}. Specifically, the model $w_\text{out}=\overline{{w}}_{ER}$ output by \name{} after $R$ rounds satisfies Eq.\eqref{eqn:conv}.

Furthermore, by the $L$-smoothness of $\Lc$ (Assumption \ref{ass:l_smooth}),
\begin{equation}
    \E\left[\Lc(w_\text{out})-\Lc(w_*)\right] \leq \frac{L}{2}\cdot \E\left[\|w_\text{out}-w_*\|^2\right] \leq \frac{L}{2}\left(A_1 + \frac{A_2}{ER+\beta}\right).
\end{equation}
\end{proof}

\subsection{Proofs of Lemmas}\label{subapp:lemma}

\begin{proof}[Proof of Lemma \ref{lemma:H2}]
This lemma is a direct consequence of Assumption \ref{ass:sample}. More specifically, for $t+1 \in \mathcal{I}_{E}=\{r E \mid r=1,2,\cdots, R\}$,
$$\mathbb{E}_{U_{t}}\left[\overline{{w}}_{t+1}\right]=\mathbb{E}_{U_{t}}\left[\frac{1}{K}\sum_{k \in U_t}v_{t+1}^{k}\right]=\frac{1}{K}\cdot K \cdot \mathbb{E}_{k\in U_{t}}\left[v_{t+1}^{k}\right]=\sum_{i\in U}p^i v^i_{t+1}=\overline{v}_{t+1},$$
where the second equality comes from the linearity of expectation, and the third equality is due to Assumption \ref{ass:sample}.
\end{proof}

\begin{proof}[Proof of Lemma \ref{lemma:H1}]
Lemma \ref{lemma:H1} is a direct consequence of Lemma 5 in \cite{li2019convergence}. The proof is outlined as follows.

For $t+1 \in \mathcal{I}_{E}=\{r E \mid r=1,2,\cdots, R\}$, $\overline{w}_{t+1}=\frac{1}{K} \sum_{k\in U_t} v_{t+1}^{k}$. Taking expectation over $U_t$,
\begin{align}\label{eqn:H1_bound_1}
    \mathbb{E}_{U_{t}}\left[\left\|\overline{w}_{t+1}-\overline{v}_{t+1}\right\|^{2}\right]=\mathbb{E}_{U_{t}}\left[\frac{1}{K^{2}}\sum_{k\in U_t}\left\|v_{t+1}^{k}-\overline{v}_{t+1}\right\|^{2}\right]&=\frac{1}{K}\mathbb{E}_{k\in U_{t}}\left[\left\|v_{t+1}^{k}-\overline{v}_{t+1}\right\|^{2}\right]\notag\\
    &=\frac{1}{K} \sum_{i\in U} p^{i}\left\|v_{t+1}^{i}-\overline{v}_{t+1}\right\|^{2}
\end{align}
where the first equality follows from Assumption \ref{ass:sample} that $\left\{v_{t+1}^{k}\right\}_{k\in U_t}$ are independent and unbiased with $\mathbb{E}_{k\in U_{t}}\left[v_{t+1}^{k}\right]=\overline{v}_{t+1}$.

To bound Eq.\eqref{eqn:H1_bound_1}, first note that since $t+1 \in \mathcal{I}_{E}$, $t_{0}:=t+1-E \in \mathcal{I}_{E}$ is also a synchronization time, which implies $\left\{w_{t_{0}}^{i}\right\}_{i\in U}$ is identical. Then,
\begin{align}\label{eqn:H1_bound_2}
\sum_{i\in U} p^{i}\left\|v_{t+1}^{i}-\overline{v}_{t+1}\right\|^{2} & =\sum_{i\in U} p^{i}\left\|\left(v_{t+1}^{i}-\overline{w}_{t_{0}}\right)-\left(\overline{v}_{t+1}-\overline{w}_{t_{0}}\right)\right\|^{2}\notag \\
&=\left(\sum_{i\in U} p^{i}\left\|v_{t+1}^{i}-\overline{w}_{t_{0}}\right\|^{2}\right) - \left\|\overline{v}_{t+1}-\overline{w}_{t_{0}}\right\|^{2} \leq \sum_{i\in U} p^{i}\left\|v_{t+1}^{i}-\overline{w}_{t_{0}}\right\|^{2},
\end{align}
where the second equality results from $\sum_{i\in U} p^{i}\left(v_{t+1}^{i}-\overline{w}_{t_{0}}\right)=\overline{v}_{t+1}-\overline{w}_{t_{0}}$. Combining Eq.\eqref{eqn:H1_bound_1} and Eq.\eqref{eqn:H1_bound_2}, we have
\begin{align}\label{eqn:H1_bound_3}
\mathbb{E}_{U_{t}}\left[\left\|\overline{w}_{t+1}-\overline{v}_{t+1}\right\|^{2}\right] &\leq \frac{1}{K} \sum_{i\in U} p^{i}\left\|v_{t+1}^{i}-\overline{w}_{t_{0}}\right\|^{2} = \frac{1}{K} \sum_{i\in U} p^{i} \left\|v_{t+1}^{i}-{w}^i_{t_{0}}\right\|^{2}\notag\\
&= \frac{1}{K} \sum_{i\in U} p^{i} \left\|\sum_{\tau=t_{0}}^{t}\eta_{\tau} g_\tau^i\right\|^{2}
\leq \frac{1}{K} \sum_{i\in U} p^{i} E \sum_{\tau=t_{0}}^{t} \left\|\eta_{\tau} g_\tau^i\right\|^{2}\notag\\
&\leq \frac{1}{K} E \sum_{\tau=t_{0}}^{t} \eta_{\tau}^{2}D^2 \leq \frac{1}{K} E^{2} \eta_{t_{0}}^{2} D^{2} \leq \frac{4}{K} \eta_{t}^{2} E^{2} D^{2}.
\end{align}
Here, the second inequality in Eq. \eqref{eqn:H1_bound_3} follows from the Cauchy-Schwarz inequality. The third inequality is a result of Assumption \ref{ass:bounded}. The fourth inequality is justified by the fact that $\eta_{t}=\frac{\alpha}{t+\beta}$ is non-increasing. Lastly, the last inequality holds since, by definition, $\beta=\max\left\{E,\frac{8L}{\mu}\right\}\geq E$, and therefore $\eta_{t_{0}} \leq 2 \eta_{t_{0}+E-1}$.
\end{proof}

\begin{proof}[Proof of Lemma \ref{lemma:H3}] First, by Eq.\eqref{eqn:w_from_v}, we have
\begin{equation}\label{eqn:F1_F2}
    \left\|\overline{{v}}_{t+1}-{w}_*\right\|^{2}=\left\|\overline{{w}}_{t}-{w}_*-\eta_{t} {g}_{t}\right\|^{2}=\left\|\overline{{w}}_{t}-{w}_*\right\|^{2} + \underbrace{\eta_{t}^{2}\left\|g_{t}\right\|^{2}}_{F_{1}} \underbrace{-2 \eta_{t}\left\langle\overline{{w}}_{t}-{w}_*, g_{t}\right\rangle}_{F_{2}}
\end{equation}

To bound $F_1$ in Eq.\eqref{eqn:F1_F2}, note that
\begin{align}\label{eqn:F1_1}
F_{1} & = \eta_{t}^{2}\left\|g_{t}\right\|^{2} = \eta_{t}^{2}\left\|\sum_{i\in U}p^i g^i_{t}\right\|^{2} \leq  \eta_{t}^{2} \sum_{i\in U}p^i\left\|g^i_{t}\right\|^{2} = \eta_{t}^{2}\sum_{i\in U}p^i\left\|G^i_t + g^i_t - G^i_t\right\|^{2}\notag\\
& = \eta_{t}^{2} \left(\sum_{i\in U}p^i\left\|G^i_t - g^i_t\right\|^{2} + 2\sum_{i\in U}p^i\left\langle G^i_t, g^i_t - G^i_t\right\rangle + \sum_{i\in U}p^i\left\|G^i_t\right\|^{2}\right)\notag\\
& = \eta_{t}^{2} \left(\sum_{i\in U}p^i\left\|G^i_t - g^i_t\right\|^{2} + 2\sum_{i\in U}p^i\left\|G^i_t\right\|\left\|g^i_t - G^i_t\right\| + \sum_{i\in U}p^i\left\|G^i_t\right\|^{2}\right)\notag\\
& \leq \eta_{t}^{2} \left(\epsilon^2 + 2\epsilon D + \sum_{i\in U}p^i\left\|G^i_t\right\|^{2}\right)\notag\\
& \leq \eta_{t}^{2} \left(\epsilon^2 + 2\epsilon D + 2L\sum_{i\in U}p^i\left(\Lc^i(w^i_t)-\Lc^i_*\right)\right).
\end{align}
Here the first inequality in Eq.\eqref{eqn:F1_1} is due to the convexity of $\|\cdot\|^2$. The second inequality comes from Assumption \ref{ass:bounded} and Assumption \ref{ass:coreset}. The last inequality follows from the fact that for $L$-smooth $\Lc^i$ (Assumption \ref{ass:l_smooth}), 
\begin{equation}\label{eqn:l_smooth}
    \left\|G^i_t\right\|^{2}\leq 2L\left(\Lc^i(w^i_t)-\Lc^i_*\right).
\end{equation}

To bound $F_2$ in Eq.\eqref{eqn:F1_F2}, note that
\begin{align}\label{eqn:F2_1}
F_{2} & =-2\eta_{t}\left\langle\overline{{w}}_{t}-{w}^{\star}, g_{t}\right\rangle = -2\eta_{t} \sum_{i\in U} p^{i}\left\langle\overline{{w}}_{t}-{w}_*, g_t^i\right\rangle\notag\\
& =-2\eta_{t} \sum_{i\in U} p^{i}\left\langle\overline{{w}}_{t}-{w}_{t}^{i}, G_t^i\right\rangle - 2\eta_{t} \sum_{i\in U} p^{i}\left\langle{w}_{t}^{i}-{w}_*, G_t^i\right\rangle + 2\eta_{t} \sum_{i\in U} p^{i}\left\langle\overline{{w}}_{t}-{w}_*, G_t^i - g_t^i\right\rangle\notag\\
& \leq -2\eta_{t} \sum_{i\in U} p^{i}\left\langle\overline{{w}}_{t}-{w}_{t}^{i}, G_t^i\right\rangle - 2\eta_{t} \sum_{i\in U} p^{i}\left\langle{w}_{t}^{i}-{w}_*, G_t^i\right\rangle + 2\eta_{t}\epsilon\cdot\left\|\overline{{w}}_{t}-{w}_*\right\|
\end{align}
Here the last inequality in Eq.\eqref{eqn:F2_1} is due to the Cauchy-Schwarz inequality and Assumption \ref{ass:coreset}. Moreover, by the Cauchy-Schwarz inequality and the AM-GM inequality, 
\begin{align}\label{eqn:F2_2}
-2\left\langle\overline{{w}}_{t}-{w}_{t}^{i}, G_t^i\right\rangle \leq \frac{1}{\eta_{t}}\left\|\overline{{w}}_{t}-{w}_{t}^{i}\right\|^{2}+\eta_{t}\left\|G_t^i\right\|^{2} \leq \frac{1}{\eta_{t}}\left\|\overline{{w}}_{t}-{w}_{t}^{i}\right\|^{2}+2L\eta_{t}\left(\Lc^i(w^i_t)-\Lc^i_*\right),
\end{align}
where the last inequality in Eq.\eqref{eqn:F2_2} follows from Eq.\eqref{eqn:l_smooth}. In addition, by the $\mu$-strong convexity of $\Lc^{i}$, Assumption \ref{ass:mu_convex},
\begin{align}\label{eqn:F2_3}
-2\left\langle{w}_{t}^{i}-w_*, G_t^i\right\rangle \leq -\left(\Lc^{i}\left({w}_{t}^{i}\right)-\Lc^{i}\left({w}_{*}\right)\right)-\frac{\mu}{2}\left\|{w}_{t}^{i}-w_*\right\|^{2}.
\end{align}
The $\mu$-strong convexity of $\Lc$, together with the optimality of $w_*$, also implies that
\begin{align}\label{eqn:F2_4}
\left\|\overline{{w}}_{t}-{w}_*\right\|\leq\frac{1}{\mu}\left\|\nabla\Lc(\overline{{w}}_{t})-\nabla\Lc({w}_*)\right\|=\frac{1}{\mu}\left\|\nabla\Lc(\overline{{w}}_{t})\right\|\leq\frac{D}{\mu},
\end{align}
where the last inequality comes from Assumption \ref{ass:bounded}.

Now combining Eq.\eqref{eqn:F1_F2} with Eq.\eqref{eqn:F1_1}, Eq.\eqref{eqn:F2_1}, Eq.\eqref{eqn:F2_2}, Eq.\eqref{eqn:F2_3}, and Eq.\eqref{eqn:F2_4}, it follows that
\begin{align}\label{eqn:F3_F4}
    \left\|\overline{{v}}_{t+1}-{w}_*\right\|^{2} &\leq \underbrace{\left\|\overline{{w}}_{t}-{w}_*\right\|^{2} - \eta_{t}\mu\sum_{i\in U}p^i\left\|{w}_{t}^{i}-w_*\right\|^{2}}_{F_3} + \sum_{i\in U}p^i\left\|\overline{{w}}_{t}-{w}_{t}^{i}\right\|^{2}\notag\\
    &\quad+ \underbrace{4L\eta_t^2\sum_{i\in U}p^i\left(\Lc^i(w^i_t)-\Lc^i_*\right) - 2\eta_t\sum_{i\in U}p^i\left(\Lc^i(w^i_t)-\Lc^i(w_*)\right)}_{F_4}\notag\\
    &\quad+ \eta_t\cdot\frac{2\epsilon D}{\mu} + \eta_t^2\cdot\left(\epsilon^2+2\epsilon D\right)
\end{align}

To bound $F_3$ in Eq.\eqref{eqn:F3_F4}, it follows by the convexity of $\|\cdot\|^2$ that
\begin{equation}\label{eqn:F3}
    F_3 \leq \left\|\overline{{w}}_{t}-{w}_*\right\|^{2} - \eta_{t}\mu\left\|\sum_{i\in U}p^i\left({w}_{t}^{i}-w_*\right)\right\|^{2}=\left(1-\eta_{t} \mu\right) \left\|\overline{{w}}_{t}-{w}_*\right\|^{2}.
\end{equation}
Meanwhile, it is shown by Lemma 1 of \cite{li2019convergence} that $F_4$ in Eq.\eqref{eqn:F3_F4} is bounded by
\begin{equation}\label{eqn:F4}
    F_4 \leq \eta_{t}^{2}\cdot6L\Gamma + \sum_{i\in U}p^i\left\|\overline{{w}}_{t}-{w}_{t}^{i}\right\|^{2}.
\end{equation}
By combining Eq.\eqref{eqn:F3_F4} with Eq.\eqref{eqn:F3} and Eq.\eqref{eqn:F4}, it follows that
\begin{align}\label{eqn:F5}
    \left\|\overline{{v}}_{t+1}-{w}_*\right\|^{2} &\leq \left(1-\eta_{t} \mu\right) \left\|\overline{{w}}_{t}-{w}_*\right\|^{2} + \eta_t\cdot\frac{2\epsilon D}{\mu} + \eta_t^2\cdot\left(6L\Gamma+\epsilon^2+2\epsilon D\right)\notag\\
    &\quad+ 2\underbrace{\sum_{i\in U}p^i\left\|\overline{{w}}_{t}-{w}_{t}^{i}\right\|^{2}}_{F_5}.
\end{align}
Finally, to bound $F_5$ in Eq.\eqref{eqn:F5}, one can apply the same argument used in bounding Eq.\eqref{eqn:H1_bound_1}. More specifically, for any $t$, there exists a $t_0 \leq t$ such that $t-t_0 \leq E-1$ and ${w}_{t_0}^{i}=\overline{w}_{t_0}$ for all $i \in U$. Then, by following the same arguments as in Eq.\eqref{eqn:H1_bound_2} and Eq.\eqref{eqn:H1_bound_3}, it is easy to verify that:
\begin{align}\label{eqn:F5_bound}
    F_5\leq \eta_{t}^{2}\cdot 4(E-1)^2D^2.
\end{align}
By plugging Eq.\eqref{eqn:F5_bound} into Eq.\eqref{eqn:F5}, we complete the proof of Lemma \ref{lemma:H3}.
\end{proof}

\newpage
\section{Evaluation Details and Extra Results}\label{app:evaluation}

\subsection{Experimental Harware and Hyper-parameters}
In our evaluations, we utilize a physical server equipped with an Intel Core X Series Core i9 10920X CPU \cite{IntelCorei910920X} and a NVIDIA GeForce RTX 2080 Ti GPU \cite{Nvidia2080Ti}. The server runs on the Linux Ubuntu 20.04 operating system. The hyper-parameters used in our evaluations are detailed in Table \ref{table: hyper-params}.

\begin{table}[ht]
  \caption{Hyper-parameters}
  \centering
  \begin{tabular}{lccc}
    \toprule
    Hyper-parameters & MNIST & Shakespeare & Synthetic \\
    \midrule
    Optimizer           & SGD & SGD & SGD \\
    Learning Rate       & 0.03 & 0.03 & 0.001 \\
    Batch Size          & 8 & 8 & 8 \\
    Local Epoch         & 10 & 10 & 10 \\
    Communication Round & 100 & 30 & 100 \\
    Number of Clients & 1000 & 143 & 30 \\
    Number of Clients per Round & 100 & 10 & 10 \\
    $\mu$ in \fedprox{} & 0.1 & 0.001 & 0.1 \\
    \bottomrule
    \label{table: hyper-params}
  \end{tabular}
\end{table}

\subsection{Extra Evaluation Results}
\begin{figure}[ht]
    \centering
    \includegraphics[width=\linewidth]{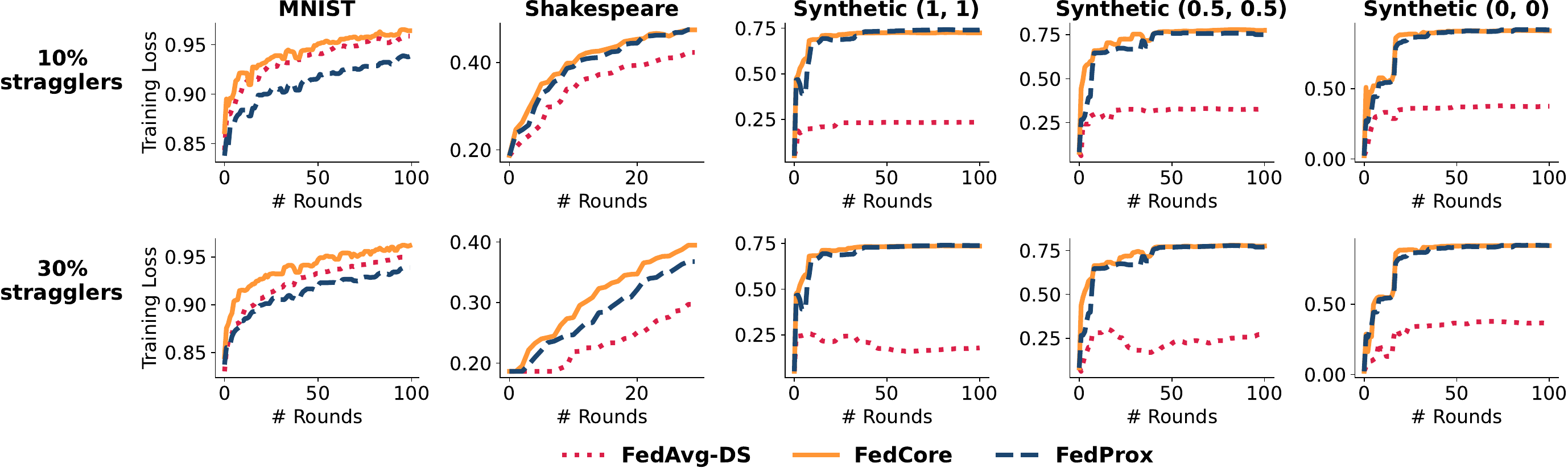}
    \caption{The test accuracy curves for \fedavgds{}, \name{}, and \fedprox{} at 10\% and 30\% stragglers.}
    \label{fig:test_acc}
\end{figure}

\begin{figure}[ht]
    \centering
    \includegraphics[width=\linewidth]{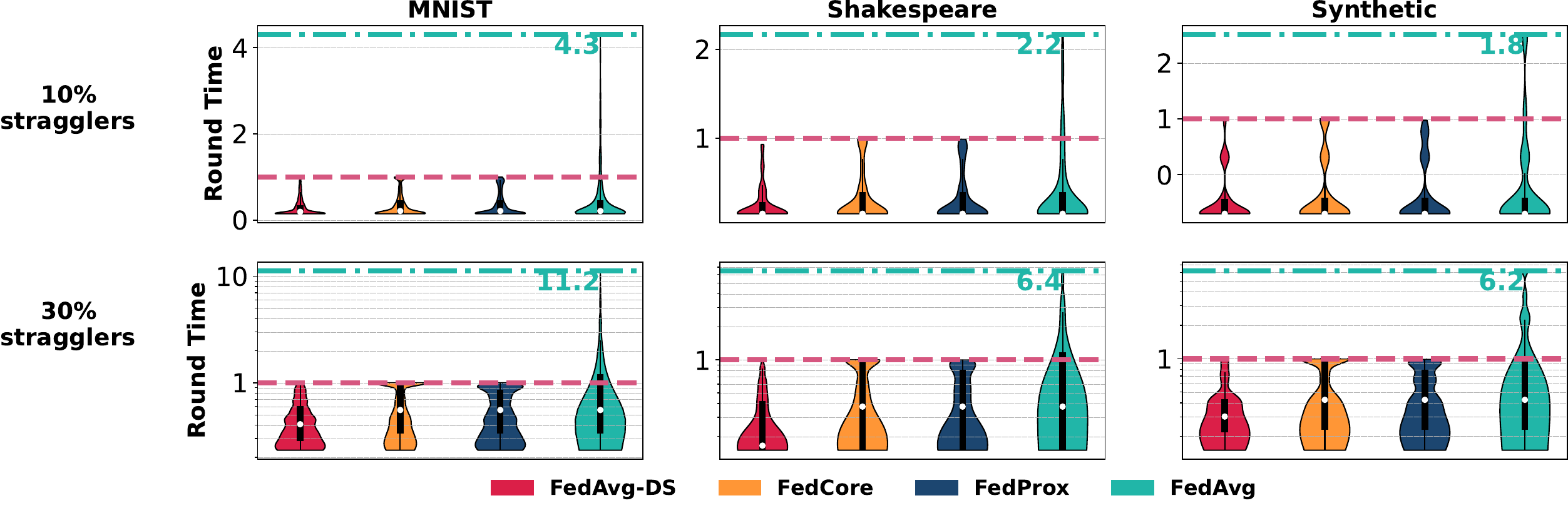}
    \caption{Round duration distribution across all clients, training rounds, and three benchmarks at 10\% and 30\% straggler settings, utilizing a log-scale y-axis for better visualization of the 30\% straggler scenario.}
    \label{fig:violin_all}
\end{figure}

\end{document}